\documentclass[11pt]{article}

\def\showauthornotes{1}
\def\showkeys{0}
\def\showdraftbox{0}

\usepackage{xspace,xcolor,enumerate}
\usepackage{color,graphicx}

\usepackage[full]{textcomp}

\usepackage{amsmath}

\usepackage{bm}

\usepackage{amsthm}

\newtheorem{theorem}{Theorem}[section]
\newtheorem*{theorem*}{Theorem}

\newtheorem*{claim*}{Claim}
\newtheorem{subclaim}{Claim}[theorem]
\newtheorem{proposition}[theorem]{Proposition}
\newtheorem*{proposition*}{Proposition}
\newtheorem{lemma}[theorem]{Lemma}
\newtheorem*{lemma*}{Lemma}
\newtheorem{corollary}[theorem]{Corollary}
\newtheorem{conjecture}[theorem]{Conjecture}
\newtheorem*{conjecture*}{Conjecture}
\newtheorem{observation}[theorem]{Observation}
\newtheorem{fact}[theorem]{Fact}
\newtheorem*{fact*}{Fact}
\newtheorem{hypothesis}[theorem]{Hypothesis}
\newtheorem*{hypothesis*}{Hypothesis}

\theoremstyle{definition}
\newtheorem{definition}[theorem]{Definition}
\newtheorem{construction}[theorem]{Construction}
\newtheorem{example}[theorem]{Example}

\newtheorem{algorithm}[theorem]{Algorithm}
\newtheorem{SDP}[theorem]{SDP}
\newtheorem{problem}[theorem]{Problem}
\newtheorem{protocol}[theorem]{Protocol}
\newtheorem{remark}[theorem]{Remark}
\newtheorem{assumption}[theorem]{Assumption}

\usepackage{fullpage}

\usepackage[varg]{txfonts} 
\renewcommand{\mathbb}{\varmathbb}

\ifnum\showkeys=1
\usepackage[color]{showkeys}
\fi

\usepackage[colorlinks,linkcolor=blue,filecolor=blue,citecolor=blue,urlcolor=blue]{hyperref}

\usepackage{prettyref}

\newcommand{\savehyperref}[2]{\texorpdfstring{\hyperref[#1]{#2}}{#2}}

\newrefformat{eq}{\savehyperref{#1}{\textup{(\ref*{#1})}}}
\newrefformat{lem}{\savehyperref{#1}{Lemma~\ref*{#1}}}
\newrefformat{def}{\savehyperref{#1}{Definition~\ref*{#1}}}
\newrefformat{thm}{\savehyperref{#1}{Theorem~\ref*{#1}}}
\newrefformat{cor}{\savehyperref{#1}{Corollary~\ref*{#1}}}
\newrefformat{cha}{\savehyperref{#1}{Chapter~\ref*{#1}}}
\newrefformat{sec}{\savehyperref{#1}{Section~\ref*{#1}}}
\newrefformat{app}{\savehyperref{#1}{Appendix~\ref*{#1}}}
\newrefformat{tab}{\savehyperref{#1}{Table~\ref*{#1}}}
\newrefformat{fig}{\savehyperref{#1}{Figure~\ref*{#1}}}
\newrefformat{hyp}{\savehyperref{#1}{Hypothesis~\ref*{#1}}}
\newrefformat{sdp}{\savehyperref{#1}{SDP~\ref*{#1}}}
\newrefformat{alg}{\savehyperref{#1}{Algorithm~\ref*{#1}}}
\newrefformat{rem}{\savehyperref{#1}{Remark~\ref*{#1}}}
\newrefformat{item}{\savehyperref{#1}{Item~\ref*{#1}}}
\newrefformat{step}{\savehyperref{#1}{step~\ref*{#1}}}
\newrefformat{conj}{\savehyperref{#1}{Conjecture~\ref*{#1}}}
\newrefformat{fact}{\savehyperref{#1}{Fact~\ref*{#1}}}
\newrefformat{prop}{\savehyperref{#1}{Proposition~\ref*{#1}}}
\newrefformat{claim}{\savehyperref{#1}{Claim~\ref*{#1}}}
\newrefformat{relax}{\savehyperref{#1}{Relaxation~\ref*{#1}}}
\newrefformat{red}{\savehyperref{#1}{Reduction~\ref*{#1}}}
\newrefformat{part}{\savehyperref{#1}{Part~\ref*{#1}}}

\newcommand{\Sref}[1]{\hyperref[#1]{\S\ref*{#1}}}

\usepackage{nicefrac}

\let\nfrac=\nicefrac

\newcommand{\half}{\nicefrac12}

\usepackage{microtype}

\ifnum\showauthornotes=1
\newcommand{\Authornote}[2]{{\sffamily\small\color{red}{[#1: #2]}}}
\newcommand{\Authorcomment}[2]{{\sffamily\small\color{gray}{[#1: #2]}}}
\newcommand{\Authorstartcomment}[1]{\sffamily\small\color{gray}[#1: }

\newcommand{\Authorfnote}[2]{\footnote{\color{red}{#1: #2}}}
\newcommand{\Authorfixme}[1]{\Authornote{#1}{\textbf{??}}}
\newcommand{\Authormarginmark}[1]{\marginpar{\textcolor{red}{\fbox{\Large #1:!}}}}
\else
\newcommand{\Authornote}[2]{}
\newcommand{\Authorcomment}[2]{}
\newcommand{\Authorstartcomment}[1]{}

\newcommand{\Authorfnote}[2]{}
\newcommand{\Authorfixme}[1]{}
\newcommand{\Authormarginmark}[1]{}
\fi

\usepackage{boxedminipage}

\newenvironment{mybox}
{\center \noindent\begin{boxedminipage}{1.0\linewidth}}
{\end{boxedminipage}
\noindent
}

\newcommand{\Brac}[1]{\left[#1\right]}

\newcommand{\abs}[1]{\left\lvert#1\right\rvert}
\newcommand{\Abs}[1]{\left\lvert#1\right\rvert}

\newcommand{\set}[1]{\left\{#1\right\}}

\newcommand{\norm}[1]{\left\lVert#1\right\rVert}

\newcommand{\iprod}[1]{\langle#1\rangle}

\newcommand{\Esymb}{\mathbb{E}}
\newcommand{\Psymb}{\mathbb{P}}

\DeclareMathOperator*{\E}{\Esymb}

\DeclareMathOperator*{\ProbOp}{\Psymb}

\renewcommand{\Pr}{\ProbOp}

\newcommand{\Ex}[2][]{\E_{{#1}}\Brac{#2}}

\usepackage{dsfont}
\usepackage{yfonts}
\usepackage{mathrsfs}

\newcommand{\tensor}{\otimes}

\newcommand{\textparen}[1]{\text{(#1)}}

\ifx\because\undefined
\newcommand{\because}[1]{\textparen{because #1}}
\else
\renewcommand{\because}[1]{\textparen{because #1}}
\fi

\newcommand{\defeq}{\stackrel{\mathrm{def}}=}     

\newcommand{\from}{\colon}

\newcommand{\mper}{\,.}
\newcommand{\mcom}{\,,}

\newcommand\bdot\bullet

\ifx\mathds\undefined 

\else

\fi

\newcommand{\etal}{et al.\xspace}

\newcommand{\sse}{\subseteq}

\newcommand{\e}{\epsilon}
\newcommand{\eps}{\epsilon}

\DeclareMathOperator{\Inf}{{\sf Inf}}

\newcommand{\sdp}{{\sf SDP }}
\DeclareMathOperator{\val}{val}

\DeclareMathOperator{\OPT}{{\sf OPT }}

\DeclareMathOperator{\vol}{{\sf vol }}
\DeclareMathOperator{\poly}{{\sf poly}}

\DeclareMathOperator{\supp}{\sf supp}

\newcommand{\Z}{\mathbb Z}
\newcommand{\N}{\mathbb N}
\newcommand{\R}{\mathbb R}

\newcommand{\problemmacro}[1]{\texorpdfstring{\textsc{#1}}{#1}\xspace}

\newcommand{\uniquegames}{\problemmacro{Unique Games}}

\newcommand{\sparsestcut}{\problemmacro{Sparsest Cut}}

\newcommand{\smallsetexpansion}{\problemmacro{Small-Set Expansion}}

\newcommand{\cA}{\mathcal A}

\newcommand{\cG}{\mathcal G} 
\newcommand{\cH}{\mathcal H}

\newcommand{\cM}{\mathcal M}
\newcommand{\cN}{\mathcal N}

\newcommand{\cP}{\mathcal P}

\newcommand{\cV}{\mathcal V}

\renewcommand{\leq}{\leqslant}
\renewcommand{\le}{\leqslant}
\renewcommand{\geq}{\geqslant}
\renewcommand{\ge}{\geqslant}

\ifnum\showdraftbox=1

\else

\fi

\let\epsilon=\varepsilon

\numberwithin{equation}{section}

\let\origparagraph\paragraph
\renewcommand{\paragraph}[1]{\origparagraph{#1.}}

\newcommand{\DSstore}[2]{%
  \global\expandafter \def \csname DSMEMORY #1 \endcsname{#2}%
}

\newcommand{\DSload}[1]{%
  \csname DSMEMORY #1 \endcsname%
}

\newcommand{\DSnewlabel}[1]{%
  \newcommand\DScurrentlabel{#1}%
  \DSoldlabel{#1}%
}

\newcommand{\DSdummylabel}[1]{}

\newcommand{\torestate}[1]{%
  \let\DSoldlabel\label%
  \let\label\DSnewlabel%
  #1%
  \DSstore{\DScurrentlabel}{#1}%
  \let\label\DSoldlabel%
}

\newcommand{\restatetheorem}[1]{%
  \let\DSoldlabel\label
  \let\label\DSdummylabel
  \begin{theorem*}[Restatement of \prettyref{#1}]
    \DSload{#1}
  \end{theorem*}
  \let\label\DSoldlabel
}

\newcommand{\restatelemma}[1]{%
  \let\DSoldlabel\label
  \let\label\DSdummylabel
  \begin{lemma*}[Restatement of \prettyref{#1}]
    \DSload{#1}
  \end{lemma*}
  \let\label\DSoldlabel
}

\newcommand{\restateprop}[1]{%
  \let\DSoldlabel\label
  \let\label\DSdummylabel
  \begin{proposition*}[Restatement of \prettyref{#1}]
    \DSload{#1}
  \end{proposition*}
  \let\label\DSoldlabel
}

\newcommand{\restatefact}[1]{%
  \let\DSoldlabel\label
  \let\label\DSdummylabel
  \begin{fact*}[Restatement of \prettyref{#1}]
    \DSload{#1}
  \end{fact*}
  \let\label\DSoldlabel
}

\newcommand{\restate}[1]{%
  \let\DSoldlabel\label
  \let\label\DSdummylabel
  \DSload{#1}
  \let\label\DSoldlabel
}

\newcommand{\ceil}[1]{\lceil #1 \rceil}

\renewcommand{\Pr}[1]{\ProbOp\Brac{#1}}
\newcommand{\pr}{\ProbOp}

\newcommand{\is}[1]{\ensuremath{\Isymb_S\Brac{#1}}}
\newcommand{\Isymb}{\mathbb{I}}

\newcommand{\yes}{\textsc{Yes}\xspace}

\newcommand{\no}{\textsc{No}\xspace}

\newcommand{\sdpval}{{\sf SDPval}}
\newcommand{\inprod}[1]{\langle #1\rangle}

\newcommand{\phiv}{\phi^{\sf V}}
\newcommand{\phivb}{\phi^{\sf V,bal}}
\newcommand{\phivs}{\Phi^{\sf V}}
\newcommand{\phivbs}{\Phi^{\sf V,bal}}

\newcommand{\fv}{\phiv}
\newcommand{\linf}{\lambda_{\infty}}
\newcommand{\phiav}[1]{\ensuremath{\Phi({#1})}}
\newcommand{\bigO}[1]{\mathcal{O}\left(#1\right) }

\newcommand{\var}{\mathsf{Var}}

\newcommand{\sound}{{\sf soundness}}
\newcommand{\Gauss}{\cG_{\Lambda, \Sigma}}
\newcommand{\Gaussp}{\cP_{\Gauss}}

\newcommand{\mset}[2]{{#1}^{\{#2\}}}

\renewcommand{\sse}{{\sf SSE}}
\newcommand{\csp}{{\sf CSP }}
\newcommand{\cvs}[2]{{\sf #1 c}-vs-{\sf s/#2}~}
\newcommand{\cvss}{{\sf c}-vs-{\sf s}~}
\newcommand{\dtv}{d_{\sf TV}}

\newcommand{\SSEH}{Small-Set Expansion Hypothesis\xspace}

\newcommand{\BalancedSeparator}{\problemmacro{Balanced Separator}}
\newcommand{\UGCexpand}{\problemmacro{Expanding Unique Games}}

\newcommand{\analyticvsep}{\problemmacro{Balanced Analytic Vertex Expansion}}

\title{The Complexity of Approximating Vertex Expansion}
\author{Anand Louis \thanks{Supported by National Science Foundation awards AF-0915903 and AF-0910584.}
	 \\Georgia Tech \\ anandl@gatech.edu \and
	 Prasad Raghavendra \thanks{ Supported by NSF Career Award and Alfred.
	P. Sloan
Fellowship } \\ UC Berkeley \\ prasad@cs.berkeley.edu \and
	Santosh Vempala \footnotemark[1] \\ Georgia Tech \\ vempala@cc.gatech.edu }

\date{}

\begin{document}

\begin{titlepage}
\maketitle
\begin{abstract}
We study the complexity of approximating the vertex expansion of graphs $G = (V,E)$, 
defined as
\[ \phiv \defeq \min_{S \subset V} n \cdot \frac{\Abs{N(S)} }{\Abs{S} \Abs{V\setminus S} }. \]

We give a simple polynomial-time algorithm for finding a subset with vertex expansion $\bigO{\sqrt{\phiv \log d}}$ where $d$ is the 
maximum degree of the graph. 
Our main result is an asymptotically matching lower bound:  under the Small Set Expansion (\sse) hypothesis, it is hard to 
find a subset with expansion less than  $C\sqrt{\phiv \log d}$ for an
absolute constant $C$.  In particular, this implies for all constant $\epsilon
> 0$, it is $\sse$-hard to distinguish whether the vertex expansion $< \epsilon$
or at least an absolute constant.  The analogous threshold for edge expansion is $\sqrt{\phi}$ with no dependence on the degree (Here $\phi$ denotes the optimal edge expansion). 
Thus our results suggest that vertex expansion is harder to approximate than edge expansion.  In particular, while Cheeger's algorithm can certify constant edge
expansion, it is \sse-hard to certify constant vertex expansion in
graphs.

Our proof is via a reduction from the {\it Unique Games} instance obtained from the \sse~ hypothesis 
to the vertex expansion problem. It involves the definition of a smoother intermediate problem we call 
\analyticvsep which is representative of both the vertex expansion and the conductance of the graph. Both reductions 
(from the UGC instance to this problem and from this problem to vertex expansion) use novel proof ideas. 
\end{abstract}

\end{titlepage}

\section{Introduction}

Vertex expansion is an important parameter associated with a
graph, one that has played a major role in both algorithms and complexity.
Given a  graph $G = (V,E)$, the vertex expansion of a set $S \subseteq
V$ of vertices is defined as 
\[ \phiv(S) \defeq \Abs{V} \cdot \frac{\Abs{N(S)}}{ \Abs{S} \Abs{V \setminus S}} \]
Here $N(S)$ denotes the outer boundary of the set $S$, i.e. 
$N(S) = \set{ i \in V \backslash S | \exists u \in S \textrm{ such that } \set{u,v} \in E }$.
The vertex expansion of the graph is given by $\phiv \defeq \min_{S \subset V} \phiv(S)$.
The problem of computing $\phiv$ is a major primitive for many graph algorithms
specifically for those that are based on the divide and conquer paradigm \cite{lr99}.
It is NP-hard to compute the vertex expansion $\phiv$ of a graph
exactly.  In this work, we study the approximability of vertex
expansion $\phiv$ of a graph.

A closely related notion to vertex expansion is that of edge
expansion. The edge expansion of a set $S$ is defined as
\[ \phi(S) \defeq \frac{ \mu(E(S,\bar{S}))}{\mu(S)} \]
and the edge expansion of the graph is $\phi = \min_{S \subset V} \phi(S)$.
Graph expansion problems have received much attention over the past decades, with applications 
to many algorithmic problems, to the construction of pseudorandom
objects and more recenlty due to their connection to the unique games
conjecture.

The problem of approximating edge or vertex expansion can be studied
at various regimes of parameters of interest.  Perhaps the simplest
possible version of the problem is to distinguish whether a given
graph is an expander.  Fix an absolute constant $\delta_0$.  A graph
is a $\delta_0$-vertex (edge) expander if its
vertex (edge) expansion is at least $\delta_0$.  The problem of recognizing a vertex expander can be stated as
follows:
\begin{problem}
Given a graph $G$, distinguish between the following two cases
\begin{description}
	\item (Non-Expander) the vertex expansion is $< \epsilon$
	\item (Expander)  the vertex expansion is $ > \delta_0$ for
		some absolute constant $\delta_0$.
\end{description}
Similarly, one can define the problem of recognizing an edge expander
graph.  
\end{problem}

Notice that if there is some sufficiently small absolute constant
$\epsilon$ (depending on $\delta_0$), for which the above problem is
easy,  then we could argue that it is easy to ``recognize'' a vertex
expander.  For the edge case, the Cheeger's inequality
yields an algorithm to recognize an edge expander.  In fact, it is possible
to distinguish a $\delta_0$ edge expander graph, from a graph whose
edge expansion is $< \delta_0^2 /2$, by just computing the second eigenvalue
of the graph Laplacian.  

It is natural to ask if there is an efficient algorithm with an
analogous guarantee for vertex expansion.  More precisely, is there
some sufficiently small $\epsilon$ (an arbitrary function of $\delta_0$), so that
one can efficiently distinguish between a graph with vertex expansion
$> \delta_0$ from one with vertex expansion $< \epsilon$.  In this
work, we show a hardness result suggesting that
there is no efficient algorithm to recognize vertex expanders.  
More precisely, our main result is a hardness for
the problem of approximating vertex expansion in graphs of bounded
degree $d$.  The hardness result shows that the approximability of
vertex expansion degrades with the degree, and therefore the problem
of recognizing expanders is hard for sufficiently large degree.
Furthermore, we exhibit an approximation
algorithm for vertex expansion whose guarantee matches the hardness
result up to constant factors. 

\paragraph{Related Work}

The first approximation for conductance was obtained by discrete
analogues of the Cheeger inequality shown
by Alon-Milman \cite{am85} and Alon \cite{a86}.  
Specifically, Cheeger's inequality relates the conductance $\phi$ to the second eigenvalue of the adjacency matrix 
of the graph -- an efficiently computable quantity.  This yields an approximation algorithm for $\phi$, one that is used heavily in practice for graph partitioning.
However, the approximation for $\phi$ obtained via Cheeger's
inequality is poor in terms
of a approximation ratio, especially when the value of $\phi$ is small.
An $\bigO{\log n}$ approximation algorithm for $\phi$ was obtained by
Leighton and Rao \cite{lr99}.  Later work by Linial \etal
\cite{llr95} and Aumann and Rabani \cite{ar98} established a strong connection between the
\sparsestcut problem and the theory of metric spaces, in turn
spurring a large and rich body of literature.  The current best
algorithm for the problem is an $O(\sqrt{\log n})$ approximation  for
due to Arora \etal \cite{arv04} using
semidefinite programming techniques. 

Amb{\"u}hl, Mastrolilli and Svensson \cite{ams07} showed that $\phiv$ and $\phi$ have no PTAS assuming that {\sf SAT} does not have  
sub-exponential time algorithms.
The current best approximation factor for $\phiv$ is $\bigO{\sqrt{\log n}}$ 
obtained using a convex relaxation \cite{fhl08}. 
Beyond this, the situation is much less clear for the approximability of vertex expansion. Applying Cheeger's method leads to a bound
of $\bigO{\sqrt{d \OPT }}$ \cite{a86} where $d$ is the maximum degree of the input graph. 

\paragraph{Small Set Expansion Hypothesis}

A more refined measure of the edge expansion of a graph is its expansion
profile.  Specifically, for a graph $G$ the expansion profile is given
by the curve 
$$ \phi(\delta) = \min_{\mu(S) \le \delta} \phi(S) \qquad \qquad
\forall \delta \in [0,\nfrac{1}{2}] \mper$$
The problem of approximating the expansion profile has received much
less attention, and is seemingly far less tractable.  In summary, the current state-of-the-art algorithms for approximating the
expansion profile of a graph are still far from satisfactory.
Specifically, the following hypothesis is consistent with the known
algorithms for approximating expansion profile.
\begin{hypothesis*}[\SSEH, \cite{rs10}]
  For every constant $\eta > 0$, there exists sufficiently small $\delta>0$
  such that given a graph $G$ it is NP-hard to distinguish the cases,
  \begin{description}\item[\yes:] 
    there exists a vertex set $S$ with volume $\mu(S)=\delta$ and expansion 
    $\phi(S)\le \eta$,
  \item[\no:]  all vertex sets $S$  with volume  $\mu(S)=\delta$ have expansion
     $\phi(S)\ge 1-\eta$.
  \end{description}
\end{hypothesis*}
Apart from being a natural optimization problem, the \smallsetexpansion problem is closely tied to the Unique
Games Conjecture.  Recent work by Raghavendra-Steurer
\cite{rs10} established
reduction from the \smallsetexpansion problem to the well known Unique
Games problem, thereby showing that \SSEH implies the Unique Games Conjecture.  This result suggests
that the problem of approximating expansion of small sets lies at the
combinatorial heart of the Unique Games problem. 

In a breakthrough work, Arora, Barak, and Steurer \cite{abs10} showed that the
problem $\smallsetexpansion(\eta,\delta)$ admits a subexponential algorithm, namely an
algorithm that runs in time $\exp(n^\eta/\delta)$.
However, such an algorithm does not refute the hypothesis that the problem
$\smallsetexpansion(\eta,\delta)$ might be hard for every constant~$\eta>0$ and
sufficiently small $\delta>0$.

The Unique Games Conjecture is not known to imply hardness results for problems closely tied to graph expansion such as \BalancedSeparator.  The reason being that the hard
instances of these problems are required to have certain global
structure namely expansion.  Gadget reductions from a unique games
instance preserve the global properties of the unique games instance
such as lack of expansion.  Therefore, showing hardness for
graph expansion problems often required a stronger version
of the \UGCexpand, where the instance is guaranteed to have good expansion.
To this end, several such variants of the conjecture for expanding graphs have
been defined in literature, some of which turned out to be false
\cite{akkstv08}.   The \SSEH could possibly serve as a natural unified
assumption that yields all the implications of expanding unique games and, in addition, also 
hardness results for other fundamental problems such as \BalancedSeparator.
In fact, Raghavendra, Steurer and Tulsiani \cite{rst12} show that the
the \sse~ hypothesis implies that the Cheeger's algorithm yields the
best approximation for the balanced separator problem.

\paragraph{Formal Statement of Results}

Our first result is a simple polynomial-time algorithm to obtain a subset of vertices $S$ whose
vertex expansion is at most $\bigO{\sqrt{\phiv \log d}}$.
Here $d$ is the largest vertex degree of $G$. The algorithm is based on a Poincair\'e-type 
graph parameter called $\lambda_\infty$ defined by Bobkov, Houdr\'e and Tetali 
\cite{bht00}, which approximates $\phiv$. While $\lambda_\infty$ also appears to be hard to compute, 
its natural \sdp relaxation gives a bound that is within $\bigO{\log d}$, as observed by Steurer and Tetali \cite{st12}, 
which inspires our first Theorem.

\begin{theorem}
\label{thm:algo}

There exists a polynomial time algorithm which given a graph $G = (V,E)$ having vertex degrees
at most $d$, outputs a set $S \subset V$, such that $ \phiv(S) = \bigO{\sqrt{ \phiv_G \log d}}$.

\end{theorem}

It is natural to ask if one can prove better inapproximability results for 
vertex expansion than those that follow from the inapproximability results for edge expansion. Indeed, the best one could hope 
for would be a lower bound matching the upper bound in the above theorem. 
Our main result is a reduction from \sse~ to the problem of distinguishing between 
the case when vertex expansion of the graph is at most $\e$ and 
the case when the vertex expansion is at least $\Omega(\sqrt{\e\log d})$. This immediately implies that
it is \sse-hard to find a subset of vertex expansion less than $C\sqrt{\phiv \log d}$ for some constant $C$.
To the best of our knowledge, 
our work is the first evidence that vertex expansion might be harder to approximate than edge expansion. 
More formally, we state our main theorem below.

\begin{theorem}
\label{thm:main}
For every $\eta > 0$, there exists an absolute constant $C$ such that $\forall \e>0 $ it is \sse-hard to distinguish 
between the following two cases for a given graph $G = (V,E)$ with maximum degree $d \geq 100/\e$.
\begin{description}
\item[\yes] : There exists a set $S \subset V$ of size $\Abs{S} \leq \Abs{V}/2$ such that 
	\[ \phiv(S) \leq \e  \]
\item[\no] : For all sets $S \subset V$, 
	\[  \phiv(S) \geq  \min \set{10^{-10}, C \sqrt{\e \log d}} - \eta \]
\end{description}
\end{theorem}

By a suitable choice of parameters in the above theorem, we obtain the
main theorem of this work, \prettyref{thm:main2}.
\begin{theorem} \label{thm:main2}
	There exists an absolute constant $\delta_0 > 0$ such that
	for every constant $\epsilon > 0$ the following holds:  Given
	a graph $G = (V,E)$, it is \sse-hard to
	distinguish between the following two cases:
	\begin{description}
\item[\yes] : There exists a set $S \subset V$ of size $\Abs{S} \leq \Abs{V}/2$ such that 
	$\phiv(S) \leq \e$

\item[\no] : ($G$ is a vertex expander with constant expansion) For all sets $S \subset V$, 
$\phiv(S) \geq  \delta_0$

\end{description}
\end{theorem}
In particular, the above result implies that it is \sse-hard to
certify that a graph is a vertex expander with constant expansion.  
This is in contrast to the case of edge
expansion, where the Cheeger's inequality can be used to certify that
a graph has constant edge expansion.

At the risk of being redundant, we note that our main theorem implies that any algorithm
that outputs a set having vertex expansion less than $C\sqrt{\phiv \log d}$ will disprove the \sse~hypothesis; 
alternatively, to improve on the bound of $\bigO{\sqrt{\phiv \log d}}$, one has to disprove the \sse~ hypothesis. 
From an algorithmic standpoint, we believe that \prettyref{thm:main2}
exposes a clean algorithmic challenge of recognizing a vertex expander
-- a challenging problem that is not only interesting on its own
right, but whose resolution would probably lead to a significant
advance in approximation algorithms.

At a high level, the proof is as follows. We introduce the notion of \analyticvsep for Markov chains. 
This quantity can be thought of as a $\csp$ on $(d+1)$-tuples of vertices. We show a reduction 
from \analyticvsep of a Markov chain, say $H$, to vertex expansion of
a graph, say $H_1$ (\prettyref{sec:sampling}). 
Our reduction is generic and works for any Markov chain $H$. 
Surprisingly, the $\csp$-like nature of \analyticvsep makes it amenable to a reduction from $\smallsetexpansion$ 
(\prettyref{sec:hardness}). 
We construct a gadget for this reduction and study its embedding into the Gaussian graph to analyze its soundness 
(\prettyref{sec:Gauss} and \prettyref{sec:gadget}).  The gadget involves a sampling procedure to generate a bounded-degree graph.

\section{Proof Overview}

\begin{figure}[htp]
\centering
\includegraphics[scale = 0.7]{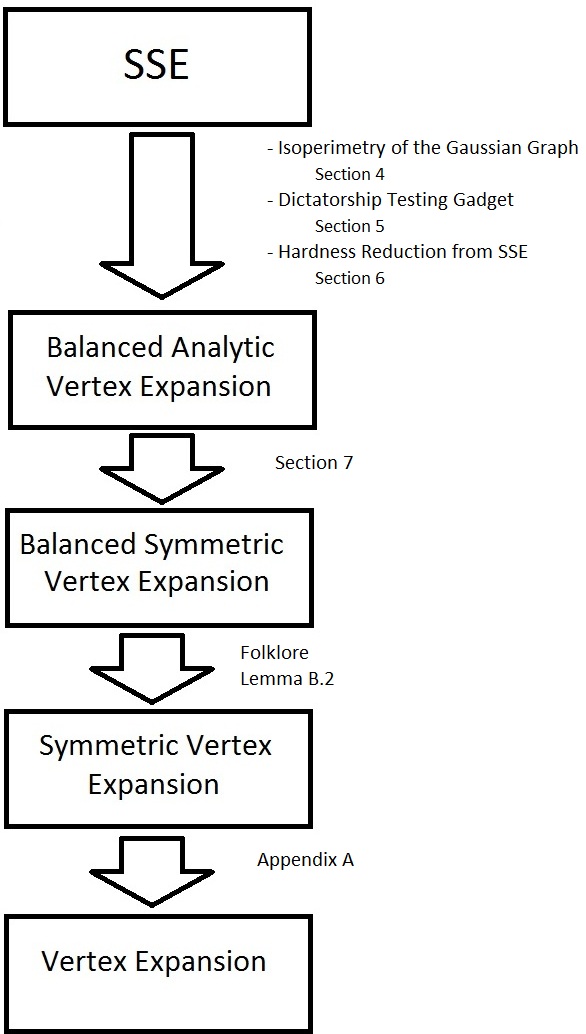}
\caption{Reduction from SSE to Vertex Expansion }
\label{fig:reductions}
\end{figure}

\paragraph{\analyticvsep}
To exhibit a hardness result, 
we begin by defining a combinatorial
optimization problem related to the problem of approximating vertex
expansion in  graphs having largest degree $d$.  This problem referred to as
\analyticvsep can be motivated as follows.

Fix a graph $G = (V,E)$ and a subset of vertices $S \subset V$.  For any
vertex $v \in V$, $v$ is on the boundary of the set $S$ if and only if
$\max_{u \in N(v)}\Abs{ \is{u} - \is{v}} = 1$, where $N(v)$ denotes
the neighbourhood of vertex $v$.  In particular, the fraction of
vertices on the boundary of $S$ is given by $\E_{v} \max_{u \in N(v)}
\Abs{ \is{u} - \is{v} }$.  The {\em symmetric} vertex expansion of the set $S \subseteq V$ is
given by,
$$ n \cdot \frac{\Abs{N(S) \cup N(V\backslash S)  }  }{ \Abs{S} \Abs{V \backslash S} } = \frac{\E_{v} \max_{u \in N(v)} \Abs{\is{u} - \is{v} }}{ \E_{u,v}
\Abs{ \is{u} - \is{v} }} \mper$$
Note that for a degree $d$ graph, each of the terms in the numerator
is maximization over the $d$ edges incident at the vertex.  The formal definition of \analyticvsep is
as shown below.

\begin{definition}

An instance of \analyticvsep, denoted by $(V,\cP)$, consists of a set of variables $V$ and a probability
distribution $\cP$ over $(d + 1)$-tuples in $V^{d+1}$. The probability distribution $\cP$ satisfies the condition 
that all its $d + 1$ marginal distributions are the same (denoted by $\mu$).
The goal is to solve the following optimization problem 
\[ \phiav{V,\cP} \defeq \min_{ F : V \to \set{0,1} |  \E_{X,Y \sim \mu} \Abs{F(X) - F(Y)} \geq \frac{1}{100}}
	 \frac{ \E_{(X,Y_1, \ldots, Y_d) \sim \cP}  \max_{i} \Abs{ F(Y_i) - F(X) }}  {\E_{X,Y \sim \mu} \Abs{F(X) - F(Y)} }  \]

\end{definition}

For constant $d$, this could be thought of as a constraint satisfaction
problem (CSP) of arity $d+1$.  
Every $d$-regular graph $G$ has an associated instance of \analyticvsep
whose value corresponds to the vertex expansion of $G$.  Conversly,
we exhibit a reduction from \analyticvsep to problem of approximating vertex
expansion in a graph of degree $\poly(d)$ (\prettyref{sec:sampling}
for details).

\paragraph{Dictatorship Testing Gadget}

As with most hardness results obtained via the label cover or
the unique games problem, central to our reduction is an appropriate dictatorship testing gadget.

Simply put, a dictatorship testing gadget for \analyticvsep is an
instance $\cH^R$ of the problem such that, on one hand there exists the
so-called {\it dictator} assignments with value $\epsilon$, while every
assignment far from every dictator incurs a cost of at least
$\Omega(\sqrt{\epsilon \log d})$.

The construction of the dictatorship testing gadget is as follows.
Let $H$ be a Markov chain on vertices $\{s,t,t',s'\}$ connected
to form a path of length three.  The transition probabilities of the Markov chain
$\cH$ are so chosen to ensure that if $\mu_H$ is the stationary
distribution of $H$ then $\mu_H(t) = \mu_H(t') = \epsilon/2$ and
$\mu_H(s) = \mu_H(s') = (1-\epsilon)/2$.  In particular, $H$ has a
vertex separator $\{t,t'\}$ whose weight under the stationary
distribution is only $\epsilon$.

The dictatorship testing gadget is over the product Markov chain
$H^R$ for some large constant $R$.  The constraints $\cP$ of the dictatorship testing gadget
$H^R$ are given by the following sampling procedure,
\begin{itemize}
	\item Sample $x \in H^R$ from the stationary distribution of
		the chain.
	\item Sample $d$-neighbours $y_1,\ldots,y_d \in H^R$ of $x$
		independently from the transition probabilities of the
		chain $H^R$.  Output the tuple $(x, y_1,\ldots,y_d)$. 
\end{itemize}

For every $i \in [R]$, the $i^{th}$ dictator solution to the above
described gadget is given by the following function,
$$F(x) = \begin{cases} 1 & \text{ if } x_i \in \{s,t\} \\ 0 & \text{
	otherwise} \end{cases}  $$
It is easy to see that for each constraint $(x,y_1,\ldots,y_d) \sim
\cP$, $\max_{j} \abs{F(x)-F(y_j)} = 0$ unless $x_i = t$ or $x_i = t'$.
Since $x$ is sampled from the stationary distribution for $\mu_H$,
$x_i \in \{t,t'\}$ happens with probability $\epsilon$.  Therefore the
expected cost incurred by the $i^{th}$ dictator assignment is at most
$\epsilon$.

\paragraph{Soundness Analysis of the Gadget}
 
The soundness property desired of the dictatorship testing gadget can
be stated in terms of influences.  Specifically, given an assignment
$F: V(H)^R \to [0,1]$, the influence of the $i^{th}$ coordinate is given
by $\Inf_i [F] = \E_{x_{[R]\backslash i}} \var_{x_i}[F(x)]$, i.e., the
expected variance of the function after fixing all but the $i^{th}$
coordinate randomly.  Henceforth, we will refer to a function $F: H^R \to [0,1]$ as {\it far from every
dictator} if the influence of all of its coordinates are small (say $<
\tau$).  

We show that the dictatorship testing gadget $H^R$ described above
satisfies the following soundness -- for every function $F$ that is far from
every dictator, the cost of $F$ is at least $\Omega(\sqrt{\epsilon
\log d})$.  To this end, we appeal to the invariance principle to
translate the cost incurred to a corresponding isoperimetric problem
on the Gaussian space.  More precisely, given a function $F : H^R \to
[0,1]$, we express it as a polynomial in the eigenfunctions over $H$.
We carefully construct a Gaussian ensemble with the same moments up to
order two, as the eigenfunctions at the query points
$(x,y_1,\ldots,y_d) \in \cP$.  By appealing to the invariance
principle for low degree polynomials, this translates in to the
following isoperimetric question over Gaussian space $\cG$.,

Suppose we have a subset $S \subseteq \cG$ of the $n$-dimensional Gaussian space.  Consider
the following experiment:
\begin{itemize}
	\item Sample a point $z \in \cG$ the Gaussian space.
	\item Pick $d$ independent perturbations $z'_1, z'_2, \ldots,
		z'_d$ of the point $z$ by
		$\epsilon$-noise.
	\item Output $1$ if at least one of the edges $(z,z'_i)$
		crosses the cut $(S,\bar{S})$ of the Gaussian space.
\end{itemize}
Among all subsets $S$ of the Gaussian space with a given volume,
which set has the least expected output in the above experiment?
The answer to this isoperimetric question corresponds to the soundness
of the dictatorship test.  A halfspace of volume $\frac{1}{2}$ has an
expected output of $\sqrt{\epsilon \log d}$ in the above experiment.
We show that among all subsets of constant volume, halfspaces acheive
the least expected output value.  

This isoperimetric theorem proven in \prettyref{sec:Gauss} yields the
desired $\Omega(\sqrt{\epsilon \log d})$ bound for the soundness of
the dictatorship test constructed via the Markov chain $H$.  Here the noise rate of $\epsilon$ arises from the fact that all the eigenfunctions
of the Markov chain $H$ have an eigenvalue smaller than $1-\epsilon$.  The details of
the argument based on invariance principle is presented in
\prettyref{sec:gadget}

We show a $\Omega(\sqrt{\epsilon \log d})$ lower bound for the
isoperimetric problem on the Gaussian space.  The proof of this
isoperimetric inequality is included in \prettyref{sec:Gauss}

We would like to point out here that the traditional noisy cube gadget
 does not suffice for our application.  This is because in the noisy
 cube gadget while the dictator solutions have an edge expansion of
 $\epsilon$ they have a vertex expansion of $\epsilon d$, yielding a
 much worse value than the soundness.

\paragraph{Reduction from \smallsetexpansion problem}
Gadget reductions from the \uniquegames problem cannot be used towards
proving a hardness result for edge or vertex expansion problems.  This
is because if the underlying instance of \uniquegames has a small
vertex separator, then the graph produced via a gadget reduction would
also have small vertex expansion.  Therefore, we appeal to a reduction from the 
\smallsetexpansion problem (\prettyref{sec:hardness} for details).

Raghavendra \etal \cite{rst12} show optimal
inapproximability results for the Balanced separator problem using a
reduction from the \smallsetexpansion problem.  While the overall approach
of our reduction is similar to theirs, the details are subtle.  Unlike
hardness reductions from unique games, the reductions for
expansion-type problems starting from \smallsetexpansion are not very well
understood.  For instance, the work of Raghavendra and Tan
\cite{rt12} gives a dictatorship testing gadget for the
Max-Bisection problem, but a \smallsetexpansion based hardness for
Max-Bisection still remains open.

\subsection{Notation}

We use $\mu_G$ to denote a probability distribution on vertices of the graph $G$. 
We drop the subscript $G$, when the graph is clear from the context. 
For a set of vertices $S$, we define $\mu(S) = \int_{x \in S} \mu(x)$.
We use $\mu_{|S}$ to denote the distribution $\mu$ restricted to the set $S \subset V(G)$. 
For the sake of simplicity, we sometimes say that vertex $v \in V(G)$ has weight 
$w(v)$, in which case we define $\mu(v) = w(v) / \sum_{u \in V} w(u)$.
We denote the weight of a set $S \subseteq V$ by $w(S)$.  
We denote the degree of a vertex $v$ by $\deg(v)$. 
We denote the neighborhood of $S$ in $G$ by $N_G(S)$, i.e. 
\[  N_G(S) = \{ v \in \bar{S} |  \exists u \in S \textrm{ such that }
	\set{u,v} \in E(G)  \}  \mper \]
We drop the subscript $G$ when the graph is clear from the context.

\subsection{Organization}

We begin with some definitions and the statements of the \sse
hypotheses in \prettyref{sec:prelims}.
In \prettyref{sec:symv}, we show that the computation of vertex expansion and 
symmetric vertex expansion is equivalent upto constant factors. 
We prove a new Gaussian isoperimetry results in \prettyref{sec:Gauss} that we use in our soundness analysis.
In \prettyref{sec:gadget} we show the construction of our main gadget and analyze its soundness and completeness 
using \analyticvsep as the test function.
We show a reduction from a reduction from \analyticvsep to
vertex expansion in \prettyref{sec:sampling}. 
In \prettyref{sec:hardness}, we use this gadget to show a reduction \sse~ to \analyticvsep.
Finally, in \prettyref{sec:putting-things-togethor}, we show how to put all the reductions togethor to 
get optimal \sse-hardness for vertex expansion.

Complimenting our lower bound, we give an algorithm that outputs a set having vertex expansion at most
$\bigO{\sqrt{\phiv \log d}}$ in \prettyref{sec:vertexsepalgo}.

\section{Preliminaries}

\label{sec:prelims}

\paragraph{Symmetric Vertex Expansion}
For our proofs, the notion of Symmetric Vertex Expansion is useful. 

\begin{definition}
Given a graph $G = (V,E)$, we define the the symmetric vertex expansion of a set 
$S \subset V$ as follows. 
\[ \phivs_G(S) \defeq n \cdot \frac{ \Abs{ N_G(S) \cup N_G(V \backslash S) }} {\Abs{S} \Abs{V \backslash {S}}}  \]

\end{definition}

\paragraph{Balanced Vertex Expansion}
We define the balanced vertex expansion of a graph as follows. 

\begin{definition}
Given a graph $G$ and balance parameter $b$, we define the {\em $b$-balanced vertex expansion} of $G$ as follows.

\[ \phivb_b \defeq \min_{ S : \Abs{S} \Abs{V \backslash S} \geq b n^2} \phiv(S).\]
and
\[ \phivbs_b \defeq \min_{ S : \Abs{S} \Abs{V \backslash S} \geq b n^2} \phivs(S).\]

We define $\phivb \defeq \phivb_{1/100}$ and $\phivbs \defeq \phivbs_{1/100}$.

\end{definition}

\paragraph{Analytic Vertex Expansion}

Our reduction from \sse~ to vertex expansion goes via an intermediate problem 
that we call $d$-\analyticvsep. 
We define the notion of $d$-\analyticvsep as follows.

\begin{definition}

An instance of $d$-\analyticvsep, denoted by $(V,\cP)$, consists of a set of variables $V$ and a probability
distribution $\cP$ over $(d + 1)$-tuples in $V^{d+1}$. The probability distribution $\cP$ satisfies the condition 
that all its $d + 1$ marginal distributions are the same (denoted by $\mu$).
The $d$-\analyticvsep under a function $F : V \to \set{0,1} $ is defined as 
\[ \phiav{V,\cP}(F)  \defeq  \frac{ \E_{(X,Y_1, \ldots, Y_d) \sim \cP}  \max_{i} \Abs{ F(Y_i) - F(X) }} 
 {\E_{X,Y \sim \mu} \Abs{F(X) - F(Y)} } \mper  \] 

The $d$-\analyticvsep of $(V,\cP)$ is defined as 
\[ \phiav{V,\cP} \defeq \min_{F : V \to \set{0,1} |  \E_{X,Y \sim \mu} \Abs{F(X) - F(Y)} \geq \frac{1}{100} } \phiav{V,\cP}(F) .  \]

When drop the degree $d$ from the notation, when it is clear from  the context.

\end{definition}

For an instance $(V,\cP)$ of \analyticvsep and
an assignment $F: V \to \set{0,1}$ define 
\[  \val_{\cP}(F) = \E_{(X,Y_1, \ldots, Y_d) \sim \cP}  \max_{i} \Abs{F(Y_i) - F(X)}. \]

\paragraph{Gaussian Graph}
Recall that two standard normal random variables $X,Y$ are said to be $\alpha$-correlated if there 
exists an independent standard normal random variable $Z$ such that 
$Y = \alpha X + \sqrt{1 - \alpha^2} Z$.

\begin{definition}
The {\em Gaussian Graph} $\Gauss$ is a complete weighted graph on the vertex set $V(\Gauss) = \R^n$. 
The weight of the edge between two vertices $u,v \in V(\Gauss)$ is given by 
\[ w(\set{u,v}) = \Pr{X = u \textrm{ and } Y = v}  \]
where $Y \sim \cN(\Lambda X, \Sigma)$, where 
$\Lambda$ is a diagonal matrix such that $\norm{\Lambda} \leq 1$ and
$\Sigma \succeq \e I$ is a diagonal matrix. 

\end{definition}

\begin{remark}
Note that for any two non-empty disjoint sets $S_1, S_2 \subset V(\Gauss)$, 
the total weight of the edges between $S_1$ and $S_2$ can be non-zero even though 
every single edge in the $\Gauss$ has weight zero.

\end{remark}

\begin{definition}
We say that a family of graphs $\cG_d$  is $\Theta(d)$-regular, if there exist absolute constants $c_1, c_2 \in \R^+$ 
such that for every $G \in \cG_d$, all vertices $i \in V(G)$ have $c_1 d \leq \deg(i) \leq c_2 d$.
\end{definition}

We now formalize our notion of hardness.

\begin{definition}
A constrained minimization problem $\cA$ with its optimal value denoted by $\val(\cA)$ is said to be \cvss hard if 
it is \sse-hard to distinguish between the following two cases. 
\begin{description} 
\item[\yes:] \[ \val(\cA) \leq c \mper \]

\item[\no:] \[ \val(\cA) \geq s \mper \]

\end{description}

\end{definition}

\paragraph{Variance}
For a random variable $X$, define the variance and $\ell_1$-variance as
follows,
$$ \var[X] = \E_{X_1,X_2} [(X_1-X_2)^2] \qquad \var_1[X] =
\E_{X_1,X_2} [|X_1 - X_2|]$$
where $X_1,X_2$ are two independent samples of $X$.

\paragraph{\SSEH}
\begin{problem}[\smallsetexpansion$(\gamma,\delta)$] 
Given a regular graph $G=(V,E)$, distinguish between the following two cases:
\begin{description}
\item[\yes:] There exists a non-expanding set $S \subset V$ with $\mu(S) = \delta$ 
and $\Phi_G(S) \leq \gamma$.
\item[\no:] All sets $S \subset V$ with $\mu(S) = \delta$ are highly expanding having 
$\Phi_G(S) \geq 1-\gamma$.
\end{description}
\end{problem}

\begin{hypothesis}[Hardness of approximating \smallsetexpansion]
For all $\gamma > 0$, there exists $\delta > 0$ such that the promise problem
\smallsetexpansion($\gamma,\delta$) is {\sf NP}-hard.
\end{hypothesis}\label{hypo:sse-hardness}

For the proofs, it shall be more convenient to use the following version of the $\smallsetexpansion$ problem,
in which we high expansion is guaranteed not only for sets of measure $\delta$, but also within 
an arbitrary multiplicative factor of $\delta$.

\begin{problem}[\smallsetexpansion$(\gamma,\delta,M)$] 
Given a regular graph $G=(V,E)$, distinguish between the following two cases:
\begin{description}
\item[\yes:] There exists a non-expanding set $S \subset V$ with $\mu(S) = \delta$ 
and $\Phi_G(S) \leq \gamma$.
\item[\no:] All sets $S \subset V$ with $\mu(S) \in
	\left(\tfrac{\delta}{M},M\delta\right)$ have
$\Phi_G(S) \geq 1-\gamma$.
\end{description}
\end{problem}

The following stronger hypothesis was shown to be equivalent to \SSEH
in \cite{rst12}.

\begin{hypothesis}[Hardness of approximating \smallsetexpansion]
For all $\gamma > 0$ and $M \geq 1$, there exists $\delta > 0$ such that the promise problem
\smallsetexpansion($\gamma,\delta,M$) is {\sf NP}-hard.
\end{hypothesis}\label{hypo:sse-hardness-main}

\section{Isoperimetry of the Gaussian Graph}
\label{sec:Gauss}

In this section we bound the \analyticvsep of the Gaussian graph.
For the Gaussian Graph, we define the canonical probability distribution on $V^{d+1}$ as follows. 
The marginal distribution along any component $X$ or $Y_i$ is the standard Gaussian distribution in 
$\R^n$, denoted here by $\mu = \cN(0,1)^n$.
\[ \Gaussp(X, Y_1, \ldots, Y_d) = \frac{ \Pi_{i =1}^d w(X,Y_i)} { \mu(X)^{d-1}  }  = 
\mu(X) \Pi_{i=1}^d \Pr{Y = Y_i}.\]
Here, random variable $Y$ is sampled from $\cN(\Lambda X, \Sigma)$.

\begin{theorem}
\label{thm:Gaussian}
For any closed set $S \subset of V(\Gauss)$ with $\Lambda$ a diagonal matrix satisfying $\norm{\Lambda} \leq 1$, 
and $\Sigma$ a diagonal matrix satisfying $\Sigma \succeq \e I$, we have 
\[  \frac{\E_{ (X,Y_1, \ldots, Y_d) \sim \Gaussp} \max_i \Abs{\is{X} - \is{Y_i} }}{\E_{X,Y \sim \mu} \Abs{\is{X} - \is{Y}}}
    =  \frac{ \E_{X \sim \mu} \E_{Y_1, \ldots Y_d \sim \cN(\Lambda X,\Sigma)} \max_{i } \abs{\is{X} - \is{Y_i}}} {\E_{X,Y \sim \mu} \Abs{\is{X} -\is{Y}}}
   \geq c\sqrt{\e \log d} \]
   for some absolute constant $c$.
\end{theorem}

\begin{lemma}
\label{lem:coupling}
Let $u,v \in \R^n$ satisfy $\Abs{u - v} \leq \sqrt{\eps \log d}$.
Let $\Lambda$ be a diagonal matrix satisfying $\norm{\Lambda} \leq 1$, 
and let $\Sigma$ a diagonal matrix satisfying $\Sigma \succeq \e I$. 
Let $P_u,P_v$ be the distributions $\cN(\Lambda u, \Sigma)$ and $\cN(\Lambda v,\Sigma)$ respectively. Then,
\[
\dtv(P_u,P_v) \leq 1 - \frac{1}{d}.
\]
\end{lemma}
\begin{proof}
First, we note that that for the purpose of estimating their total variation distance, 
we can view $P_u,P_v$ as one-dimensional Gaussians along the line $\Lambda u- \Lambda v$. 
Since $\norm{\Lambda} \leq 1$, 
\[ \norm{ \Lambda u - \Lambda v } \leq \norm{u - v} \leq \sqrt{\e \log d} \mper \]
Wlog, we may
take $\Lambda u=0$ and $\Lambda v = \sqrt{\eps\log d}$. Next, by the definition of total variation distance, 
\begin{eqnarray*}
\dtv(P_u, P_v) &=& \int_{x: P_v(x) \ge P_u(x)} |P_v(x) - P_u(x)|dx \\
&=& \int_{\Lambda v/2}^{\infty} (P_v(x) - P_u(x))dx\\ 
&=& \frac{1}{\sqrt{2\pi\eps}} \int_{\Lambda v/2}^\infty e^{-\frac{\|x- \Lambda v\|^2}{2\eps}}\, dx  -
	 \frac{1}{\sqrt{2\pi\eps}} \int_{\Lambda v/2}^\infty e^{-\frac{\|x\|^2}{2\eps}} \, dx\\
&=& \frac{1}{\sqrt{2\pi\eps}} \int_{-\Lambda v/2}^{\Lambda v/2} e^{-\frac{\|x\|^2}{2\eps}}\, dx\\
&=& \frac{1}{\sqrt{2\pi}} \int_{-\sqrt{\log d}/2}^{\sqrt{\log d}/2} e^{-\frac{\|x\|^2}{2}}\, dx\\
&=& 1 -   2\cdot \frac{1}{\sqrt{2\pi}} \int_{\sqrt{\log d}/2}^\infty e^{-\frac{\|x\|^2}{2}}\, dx\\
&<& 1 - \frac{1}{d}.
\end{eqnarray*}
where the last step uses a standard bound on the Gaussian tail.

\end{proof}

\begin{proof}[Proof of \prettyref{thm:Gaussian}.]
Let $\mu_X$ denote the Gaussian distribution $\cN(\Lambda X,\Sigma)$. Then the LHS is:
\[
\int_{\R^n\setminus S} \left(1-(1-\mu_X(S))^d\right) \, d\mu(X) + \int_{S}\left(1
- (1-\mu_X(\R^n\setminus S))^d\right) \, d\mu(X).
\]
To bound this, we will restrict ourselves to points $X$ for which the $\mu_X$ measure of 
the complementary set is at least $1/d$. Roughly speaking, these will be points near the boundary of $S$.
Define:
\[ S_1 = \set{ x \in S \, : \, \mu_X (\R^n \setminus S) < \frac{1}{2d} },\ 
S_2 = \set{ x \in \R^n \setminus S\, : \, \mu_X (S) < \frac{1}{2d} } \]
and
\[ S_3 = \R^n \setminus S_1 \setminus S_2. \]
For $u \in \R^n$, let $P_u$ be the distribution $\cN(\Lambda u,\Sigma)$.
For any $u \in S_1, v \in S_2$, we have
\[
\dtv(P_u, P_v) > 1 - \frac{1}{2d} - \frac{1}{2d} = 1- \frac{1}{d}.
\]
Therefore, by \prettyref{lem:coupling}, $ \|u - v\| > \sqrt{\eps\log d}$,
i.e., $d(S_1,S_2) > \sqrt{\eps\log d}$. 
Next we bound the measure of $S_3$. 
We can assume wlog that $\mu(S) \le \mu(\R^n\setminus S)$ and $\mu(S_1) \ge \mu(S)/2$ (else  $\mu(S_3) \ge \mu(S)/2$ and we are done).
Applying the isoperimetric inequaity for Gaussian space \cite{b75,st78}, for subsets at this distance,
\[
\mu(S_3) \ge \sqrt{\frac{2}{\pi}} \sqrt{\eps\log d}\cdot \mu(S_1) \mu(S_2) \ge \sqrt{\frac{\eps \log d}{2\pi}}\cdot \mu(S)\mu(\R^n \setminus S).
\]
We are now ready to complete the proof.
\begin{eqnarray*}
&&\frac{1}{2}\left(\int_{\R^n\setminus S} (1-(1-\mu_X(S))^d) \, d\mu(X) + \int_{S}(1 - (1-\mu_X(\R^n\setminus S)) \, d\mu(X)\right) \\
&\ge& \frac{1}{2} \left(\int_{X \in \R^n\setminus S, \mu_X(S) \ge 1/d}
(1-(1-\mu_X(S))^d) \, d\mu(X) + \int_{X \in S, \mu_X(\R^n\setminus S) \ge 1/d}(1 - (1-\mu_X(\R^n\setminus S)) \, d\mu(X)\right) \\
&\ge& \frac{e-1}{2e} \left(\int_{X \in \R^n\setminus S, \mu_X(S) \ge 1/d}  \, d\mu(X) + \int_{X \in S, \mu_X(\R^n\setminus X) \ge 1/d} \, d\mu(X)\right)\\
&\ge&\frac{e-1}{2e}\mu(S_3)\\
&\ge& c\sqrt{\eps\log d}\cdot \mu(S)\mu(\R^n\setminus S).
\end{eqnarray*} 
\end{proof}

We prove the following Theorem which helps us to bound the isoperimetry of the Gaussian graph for 
over all functions over the range $[0,1]$. 

\begin{theorem}
\label{thm:cont-to-bin}
Given an instance $(V,\cP)$ and a function $F : V \to [0,1]$,
 there exists a function $F' : V \to \set{0,1}$, such that
\[ \frac{\E_{ (X,Y_1, \ldots, Y_d) \sim \cP }  \max_i \Abs{ F(X) - F(Y_i)} }{ \E_{ X,Y \sim \mu}  \Abs{ F(X) - F(Y) } }  
\geq 
\frac{\E_{ (X,Y_1, \ldots, Y_d) \sim \cP }  \max_i \Abs{ F'(X) - F'(Y_i)} }{ \E_{ X,Y \sim \mu}  \Abs{ F'(X) - F'(Y) } }     \]
\end{theorem}

\begin{proof}
For every $r \in [0,1]$, we define $F_r : V \to \set{0,1}$ as follows.
\[ F_r(X) = \begin{cases}   1 & F(X) \geq r \\ 0 & F(X) < r     \end{cases}  \]
Clearly, 
\[ F(X) = \int_0^1 F_r(X) dr \mper  \]
Now, observe that if $F(X) - F(Y) \geq 0$ then $F_r(X) - F_r(Y) \geq 0\ \forall r \in [0,1]$ 
and similiarly, if $F(X) - F(Y) < 0$ then $F_r(X) - F_r(Y) \leq 0\ \forall r \in [0,1]$.
Therefore, 
\[  \Abs{F(X) - F(Y) } = \Abs{ \int_0^1 \left( F_r(X) - F_r(Y) \right) dr } = \int_0^1 \Abs{F_r(X) - F_r(Y)} dr  \mper    \]
Also, observe that if $\Abs{ F(X) - F(Y_1) }  \geq \Abs{F(Y_i) - F(X) }$ then 
\[   \Abs{ F_r(X) - F_r(Y_1) }  \geq \Abs{F_r(Y_i) - F_r(X) }\  \forall r \in [0,1]  \]

Therefore, 
\begin{eqnarray*}
\frac{\E_{ (X,Y_1, \ldots, Y_d) \sim \cP }  \max_i \Abs{ F(X) - F(Y_i)} }{ \E_{ X,Y \sim \mu}  \Abs{ F(X) - F(Y) } }  
& = & 
\frac{\E_{ (X,Y_1, \ldots, Y_d) \sim \cP }  \max_i \int_0^1 \Abs{ F_r(X) - F_r(Y_i)} dr}
		{ \E_{ X,Y \sim \mu} \int_0^1 \Abs{ F_r(X) - F_r(Y)} dr }  \\
&=& \frac{\int_0^1 \left( \E_{ (X,Y_1, \ldots, Y_d) \sim \cP }  \max_i  \Abs{ F_r(X) - F_r(Y_i)} \right) dr}
		{ \int_0^1 \left( \E_{ X,Y \sim \mu} \Abs{ F_r(X) - F_r(Y)} \right) dr }  \\
& \geq & \min_{r \in [0,1] }
\frac{ \E_{ (X,Y_1, \ldots, Y_d) \sim \cP }  \max_i  \Abs{ F_r(X) - F_r(Y_i)}}{\E_{ X,Y \sim \mu} \Abs{ F_r(X) - F_r(Y)} }  \\
\end{eqnarray*}
Let $r'$ be the value of $r$ which minimizes the expression above. Taking $F'$ to be $F_{r'}$
finishes the proof.

\end{proof}

\begin{corollary}[Corollary to \prettyref{thm:Gaussian} and \prettyref{thm:cont-to-bin}]
\label{cor:Gaussian}
Let $F : V(\Gauss) \to [0,1]$ be any function. Then,    for some absolute constant $c$,  
\[  \frac{\E_{ (X,Y_1, \ldots, Y_d) \sim \Gaussp} \max_i \Abs{F(X) - F(Y_i) }}{\E_{X,Y \sim \mu} \Abs{F(X) - F(Y)}}
   \geq c\sqrt{\e \log d} \mper \]
\end{corollary}

%
%

\section{Dictatorship Testing Gadget}
\label{sec:gadget}

In this section we initiate the construction of the dictatorship
testing gadget for reduction from \sse.

Overall, the dictatorship testing gadget is obtained by picking an
appropriately chosen constant sized Markov-chain $H$, and considering
the product Markov chain $H^R$.  Formally, given a Markov chain $H$,
define an instance of $\analyticvsep$ with vertices as $V_H$ and the
constraints given by the following canonical probability distribution over
$V_H^{d+1}$.
\begin{itemize}
	\item Sample $X \sim \mu_H$, the stationary distribution of
		the Markov chain $V_H$.
	\item Sample  $Y_1,\ldots, Y_d$ independently from the
		neighbours of  $X$ in $V_H$
\end{itemize}

For our application, we use a specific Markov chain $H$ on four
vertices. 
Define a Markov chain $H$ on $V_H = \{s,t,t',s'\}$ as follows,$
p(s|s) =
p(s'|s') = 1-\frac{\epsilon}{1-2\epsilon}$, $p(t|s) = p(t'|s') =
\frac{\epsilon}{1-2\epsilon}$, $p(s|t) = p(s'|t') =
\frac{1}{2}$ and $p(t'|t) = p(t|t') = \frac{1}{2}$.  It is easy to see
that the stationary distribution of the Markov chain $H$ over $V_H$ is
given by,
\[  \mu_H(s) = \mu_H(s') = \frac{1}{2} - \e \qquad \qquad \mu_H(t) = \mu_H(t') = \e \]
From this Markov chain, construct a dictatorship testing gadget
$(V_H^R, \cP_{H}^R)$ as described above.   We begin by showing that
this dictatorship testing gadget has small vertex separators
corresponding to dictator functions.
%
\begin{proposition}[Completeness]
\label{prop:completeness}
For each $i \in [R]$, the $i^{th}$-dictator set defined as $F(x) = 1$
if $x_i \in \{s,t\}$ and $0$ otherwise satisfies,
$$ \var_1[F] = \frac{1}{2} \qquad \text{ and } \qquad
\val_{\cP_{H^R}}(F) \leq 2\epsilon
$$
%
\end{proposition}
\begin{proof}
Clearly, 
\[   \E_{X,Y \sim \mu_H}  \Abs{F(X) - F(Y)} = \frac{1}{2}  \]
Observe that for any choice of $(X, Y_1, \ldots, Y_d) \sim \cP_{H^R}$,
$\max_i \Abs{F(X) - F(Y_i)}$ is non-zero if and only if either $x_i =
t$ or $x_i = t'$.  Therefore we have,
\[ \E_{ (X, Y_1, \ldots, Y_d) \sim \cP_H } \max_i \Abs{F(X) - F(Y_i)} 
\leq \pr[x_i \in \{t,t'\}]) =  2\e \mcom \]
which concludes the proof.
%
\end{proof}

\subsection{Soundness}
We will show a general soundness claim that holds for dictatorship
testing gadgets $(V(H^R),\cP_{H^R})$ constructed out of arbitrary
Markov chains $H$ with a given spectral gap.  Towards formally stating
the soundness claim, we recall some background and notation about
polynomials over the product Markov chain $H^R$.

\subsection{Polynomials over $H^R$}
	In this section, we recall how functions over the product
	Markov chain $H^R$ can be written as multilinear polynomials
	over the eigenfunctions of $H$.

  Let $e_0, e_1, \ldots, e_{n} : V(H) \to \R$ be an orthonormal
  basis of eigenvectors of $H$ and let $\lambda_0, \ldots, \lambda_n$ be the corresponding eigenvalues.
  Here $e_0 = 1$ is the constant function whose eigenvalue $\lambda_0
  = 1$.  Clearly $e_0,\ldots, e_n$ form an orthonormal basis for the vector
space of functions from $V(H)$ to $\R$.  

  It is easy to see that the eigenvectors of the product chain $H^R$
  are given by products of $e_0,\dots, e_n$.  Specifically, the
  eigenvectors of $H^R$ are indexed by $\sigma \in [n]^R$ as follows,
  $$ e_{\sigma}(x) = \prod_{i = 1}^R e_{\sigma_i}(x_i)$$
  Every function $f: H^R \to \R$ can be written in this orthonormal basis
  $ f(x)= \sum_{\sigma \in [n]^R} \hat{f}_{\sigma} e_{\sigma}(x)$.
  For a multi-index $\sigma \in [n]^R$, the function $e_{\sigma}$ is
   a monomial of degree $|\sigma| = |\{i |
  \sigma_i \neq 0\}|$.  
  
  For a polynomial $Q = \sum_{\sigma} \hat{Q}_{\sigma} e_{\sigma}$,
  the polynomial $Q^{> p}$ denotes the projection on to degrees higher
  than $p$, i.e., $Q^{> p} = \sum_{\sigma, |\sigma|>p}
  \hat{Q}_{\sigma} e_{\sigma}$.
  The influences of a polynomial $Q = \sum_{\sigma} \hat{Q}_{\sigma}$
  are defined as,
  $$ \Inf_i(Q) = \sum_{\sigma: \sigma_i \neq 0}
  \hat{Q}_{\sigma}^2$$
  The above notions can be naturally extended to vectors of
  multilinear polynomials $Q = (Q_0, Q_1, \ldots, Q_d)$.

  Note that every real-valued function on the vertices $V(H)$ of a Markov chain
  $H$ can be thought of as a random variable.  For each $i > 0$, the
  random variable $e_i(x)$ has mean zero and variance $1$.  The same
  holds for all $e_{\sigma}(x)$ for all $|\sigma| \neq 0$.  For a
  function $Q : V(H^R) \to \R$ (or equivalently a polynomial),
  $\var[Q]$ denotes the variance of the random variable $Q(x)$ for a
  random $x$ from stationary distribution of $H^R$.  It is an easy
  computation to check that this is given by,
  $$ \var[Q] = \sum_{\sigma: |\sigma| \neq 0}
  \hat{Q}_{\sigma}^2$$
 
  We will make use of the following Invariance Principle due to Isaksson and Mossel \cite{im12}.
\begin{theorem}[\cite{im12}]
\label{thm:invar}
Let $X = (X_1, \ldots ,X_n)$ be an independent sequence of ensembles, 
such that $\Pr{X_i = x} \geq \alpha > 0, \forall i, x$. Let $Q$ be a $d$-dimensional multilinear polynomial such that 
$\var(Q_j(X)) \leq 1$,  $\var(Q_j^{>p}) \leq (1 - \e \eta)^{2p}$  and
$\Inf_i (Q_j) \leq \tau$ where $p = \frac{1}{18} \log (1 / \tau)/\log (1/\alpha)$. 
Finally, let $\psi : \R^k \to \R$ be Lipschitz continuous. Then,
\[ \Abs{  \Ex{ \psi(Q(X) )}  - \Ex{\psi(Q(Z))} } = \bigO{ \tau^{ \frac{\e \eta}{18}/ \log \frac{1}{\alpha} } }  \]
where $Z$ is an independent sequence of Gaussian ensembles with the same covariance structure as $X$.

\end{theorem}

\subsection{Noise Operator}
We define a noise operator $\Gamma_{1 - \eta}$ on functions on the
Markov chain $H$ as follows :  
\[ \Gamma_{1 - \eta}F (X) \defeq (1 - \eta) F(X) + \eta \E_{Y \sim X} F(Y)   \]
for every function $F: H \to \R$.  Similarly, one can define the noise
operator $\Gamma_{1-\eta}$ on functions over $H^R$.

Applying the noise operator
$\Gamma_{1-\eta}$ on a function $F$, smoothens the function or makes
it closer to a low-degree polynomial.  This resulting function
$\Gamma_{1-\eta} F$ is close to a {\it low-degree polynomial}, and
therefore is amenable to applying an invariance principle.  Formally,
one can show the following decay of coefficients of high degree for
$\Gamma_{1-\eta} F$. We defer the proof to the Appendix (\prettyref{lem:lowdeg}).
\begin{lemma}
\label{lem:lowdeg1}
(Decay of High degree Coefficients)
Let $Q_j$ be the multi-linear polynomial representation of $\Gamma_{1
- \eta}F(X)$, and let $\e$ be the spectral gap of the Markov chain $H$.
Then, 
\[ \var(Q_j^{>p}) \leq (1 - \e \eta)^{2p} \]
\end{lemma}
Furthermore, on applying the noise operator $\Gamma_{1-\eta}$, the
resulting function $\Gamma_{1-\eta} F$ can have a bounded number of
influential coordinates as shown by the following lemma.
\begin{lemma}(Sum of Influences Lemma)
\label{lem:sum-of-influences1}
	If the spectral gap of a Markov chain is at least $\e$ then for any function $F: V_H^R \to \R$, 
	$$ \sum_{i \in [R]} \Inf_i(\Gamma_{1-\eta} F) \leq
	\frac{1}{\eta \epsilon} \var[F] $$
\end{lemma}

\begin{proof}
	By suitable normalization, we may assume without loss of
	generality that $\var[F] = 1$.
	If $Q$ denotes the multilinear representation of
	$\Gamma_{1-\eta} F$, then the
	sum of influences can be written as,
	\begin{align*}
	\sum_{i \in [R]} \Inf_i(\Gamma_{1-\eta} F) 
	& \leq \sum_{|\sigma| \neq 0 } |\sigma| \hat{Q}_{\sigma}^2 \\
	& \leq \sum_{|\sigma| \neq 0 } |\sigma| (1-\eta\epsilon)^{2
	|\sigma|} \hat{F}_{\sigma}^2 \\
	& \leq \left(\max_{k \in \N} k (1-\eta\epsilon)^{2k}\right)
	\sum_{|\sigma| \neq 0} \hat{F}_{\sigma}^2 
	 < \frac{1}{\eta \epsilon} 
\end{align*}
where we used the fact that the function $h(t) = t (1-\eta
\epsilon)^{2t}$ achieves its maximum value at $t = -\frac{1}{2}
\ln(1-\eta \epsilon)$.
\end{proof}

\subsection{Soundness Claim}
Now we are ready to formally state our soundness claim for a
dictatorship test gadget constructed out of a Markov chain.  
\begin{proposition}[Soundness]
\label{prop:soundgadget}
For all $\epsilon, \eta, \alpha, \tau > 0$ the following holds. 
Let $H$ be a finite Markov-chain with a spectral gap of at least $\e$, and the probability of every
state under
stationary distribution is $\geq \alpha$.
Let $F : V(H^R) \to \set{0,1}$ be a function such that $\max_{i
\in [R]} \Inf_i(\Gamma_{1-\eta} F)
\leq \tau$. Then we have 
\[ \E_{ (X,Y_1, \ldots, Y_d) \sim \cP_{H^R}} [ \max_{i}
	\Abs{F(Y_i) - F(X)}] \geq  \Omega( \sqrt{\e \log d }) \E_{X,Y
		\sim \mu_{H^R}} \Abs{F(X) - F(Y)} - O(\eta) - \tau^{\Omega(\eps\eta/\log(1/\alpha))} \]
\end{proposition}
For the sake of brevity, we define $\sound(V(H^R),\cP_{H^R}) $ to be the  following : 
\begin{definition}
	\[ \sound(V(H^R),\cP_{H^R}) \defeq \min_{F : \max_{i\in [R]} \Inf_i(F) \leq \tau  }  
 \frac{\E_{ (X,Y_1, \ldots, Y_d) \sim \cP_{H^R}} [ \max_{i}
	 \Abs{F(Y_i) - F(X)}]}{ \E_{X,Y \sim \mu_{H^R})} \Abs{F(X) - F(Y)}}
 \]
\end{definition}
In the rest of the section, we will present a proof of
\prettyref{prop:soundgadget}.  First, we construct gaussian random
variables with moments matching the eigenvectors of the chain $H$.


%

\paragraph{Gaussian Ensembles}

 Let $Q = (Q_0, Q_1,\ldots, Q_d)$ be the multi-linear polynomial
representation of the vector-valued function $\left(\Gamma_{1 - \eta} F(X), \Gamma_{1-\eta} F(Y_1), \ldots, \Gamma_{1 - \eta} F(Y_d) \right)$.
Let $E$ denote the ensemble of $nd$ random variables $(e_0(X),
e_1(X),\ldots, e_n(X)), (e_{0}(Y_1),\ldots, e_n(Y_1)),\ldots,
(e_{0}(Y_d),\ldots, e_n(Y_d))$.  
Let $E_1,\ldots, E_R$ be $R$ independent copies of the ensemble $E$.
Clearly, the polynomial $Q$ can be thought of as a  polynomial over
$E_1,\ldots, E_R$.  For each random variable $x$ in $E_1,\ldots, E_R$ and a value $\beta$ in
its support, $\Pr{x = \beta}$ is at least the minimum probability of a
vertex in $H$ under its stationary distribution.

This polynomial $Q$ satisfies the requirements of \prettyref{thm:invar} because
on the one hand, the influences of $F$ are $\leq \tau$ and on the
other by \prettyref{lem:lowdeg1}, $\var(Q^{\geq p}) \leq (1-\e
\eta)^{2p}$.  Now we will apply the invariance principle to relate the
soundness to the corresponding quantity on the gaussian graph, and
then appeal to the isoperimetric result on the Gaussian graph
(\prettyref{thm:Gaussian}).

The invariance principle translates the polynomial
$(Q_0(X),Q_1(Y_1),\ldots Q_d(Y_d))$ on the sequence of
independent ensembles $E_1,\dots, E_R$, to a polynomial on a corresponding sequence of
gaussian ensembles with the same moments up to degree two. 

Consider the ensemble $E$.  For each $i \neq 0$, the expectation
$\E[e_i(X)] = \E[e_i(Y_1) = 0] = \ldots \E[e_i(Y_d)] =0$.  For each
$i \neq j$, it is easy to see that,
$\E[e_i(X) e_j(X)] = \E[e_i(Y_1)e_j(Y_1)] =  \ldots
\E[e_i(Y_d)e_j(Y_d] = 0$.  Moreover, $\E[e_i(X) e_j(Y_a)] =
\E[e_{i}(Y_a) e_j(Y_b)] = 0$ whenever $i \neq j$ and all $a,b \in \{1,\ldots d\}$.
The only non-trivial correlations are $\E[e_{i}(X)e_i(Y_a)]$ and
$\E[e_i(Y_a)e_i(Y_b)]$ for all $i \in [n]$ and $a,b \in [d]$.  It is
easy to check that 
$$ \E[e_i(X) e_i(Y_a)] = \lambda_i  \qquad \qquad \E[e_i(Y_a)
e_i(Y_b)] = \lambda_i^2$$
From the above discussion, we see that the following gaussian ensemble
$z = (z_X, z_{Y_1}, \ldots, z_{Y_d})
$ has the same covariance as the ensemble $E$.

\begin{enumerate}
\item Sample $z_X$  and $n$-dimensional Gaussian random vector.

\item Sample $z_{Y_1}, \ldots, z_{Y_d} \in \R^n$ i.i.d as follows :
 The $i^{th}$ coordinate of each $z_{Y_a}$ is sampled from $ \lambda_i
 z_X(i) + \sqrt{1-\lambda_i^2} \xi_{a,i}$ where $\xi_{a,i}$ is a Gaussian random variable
independent of $z_X$ and all other $\xi_{a,i}$.
\end{enumerate}
Let $Z_X, Z_{Y_1},\ldots, Z_{Y_d} \in \R^{nR}$ be the ensemble
obtained by $R$ independent samples from $z_{X}, z_{Y_1},\ldots,
z_{Y_d}$.

Let $\Sigma$ denote the $nR \times nR$ diagonal matrix whose entries
are $1-\lambda^2_1,
\ldots, 1-\lambda^2_n$ repeated $R$ times.  Since the spectral gap of
$H$ is $\e$, we have that $1 - \lambda^2_{i} \geq 2\e-\epsilon^2 >
\epsilon$ for all  $i \in \{1,\ldots, n\}$.  Therefore, we have
$\Sigma > \epsilon I$.

%
%

\paragraph{Proof of soundness}

Now we return to the proof of the main soundness claim for the
dictatorship testing gadget ($V(H^R)$, $\cP_{\cH^R}$) constructed out
an arbitrary Markov chain.
\begin{proof}[Proof of \prettyref{prop:soundgadget}]
 Let $Q = (Q_0, Q_1,\ldots, Q_d)$ be the multi-linear polynomial
representation of the vector-valued function $\left(\Gamma_{1 - \eta} F(X), \Gamma_{1-\eta} F(Y_1), \ldots, \Gamma_{1 - \eta} F(Y_d) \right)$.

Define a function $s : \R \to \R$ as follows
  $$ s(x) = \begin{cases} 0 & \text{ if } x < 0 \\ x & \text{ if } x
	  \in [0,1] \\ 1 & \text{ if } x > 1 \end{cases} $$
  Define a function $\Psi : \R^{d+1} \to \R$ as, $\Psi(x,y_1,\ldots,
  y_d) = \max_{i} |s(y_i) - s(x)|$.  Clearly, $\Psi$ is a Lipshitz function
  with a constant of 1.

  Using the fact that $F$ is bounded in $[0,1]$,
  \begin{eqnarray} \label{eq:num1}
\E_{(X,Y_1, \ldots, Y_d) \sim \cP_{H^R} } \max_a  \Abs{F(X) - F(Y_a) }
\geq 
\E_{(X,Y_1, \ldots, Y_d) \sim \cP_{H^R} } \max_a  \Abs{\Gamma_{1 -
\eta}F(X) - \Gamma_{1 - \eta}F(Y_a) } - 2\eta
  \end{eqnarray}	 
  Furthermore, since $\Gamma_{1-\eta} F$ is also bounded in $[0,1]$,
  we have $s(\Gamma_{1-\eta} F) = \Gamma_{1-\eta} F$.  Therefore,
  \begin{eqnarray}
\E_{(X,Y_1, \ldots, Y_d) \sim \cP_{H^R} } \max_a  \Abs{\Gamma_{1 -
\eta}F(X) - \Gamma_{1 - \eta}F(Y_a) } = \E_{(X,Y_1, \ldots, Y_d) \sim
	\cP_{H^R} } \max_a  \Abs{s \left( \Gamma_{1 -
\eta}F(X)\right) - s\left(\Gamma_{1 - \eta}F(Y_a)\right) } 
  \end{eqnarray}	 
  Apply the invariance principle to the polynomial $Q =
  \left(\Gamma_{1 - \eta} F, \Gamma_{1-\eta} F, \ldots, \Gamma_{1 -
  \eta} F \right)$ and Lipshitz function $\Psi$.  
  By invariance principle \prettyref{thm:invar}, we get
\begin{align}
\E_{(X,Y_1, \ldots, Y_d) \sim \cP_{H^R} } \max_a
 & \Abs{s\left(\Gamma_{1 - \eta}F(X)\right) - s\left(\Gamma_{1 -
\eta}F(Y_a)\right) } \nonumber \\ 
& \geq  \E_{(Z_X ,Z_{Y_1}, \ldots, Z_{Y_d}) \sim \Gaussp} \max_a
 \Abs{s\left(\Gamma_{1 - \eta}F(Z_X)\right) - s\left(\Gamma_{1 - \eta}
 F(Z_{Y_a})\right)} -
 \tau^{\Omega(\epsilon\eta/\log(1/\alpha))} 
\end{align}
Observe that $s \circ (\Gamma_{1-\eta} F)$ is bounded in $[0,1]$ even
over the gaussian space.  Hence, by using the isoperimetric result on gaussian graphs
(\prettyref{cor:Gaussian}), we know that
\begin{eqnarray}
\E_{(Z_X ,Z_{Y_1}, \ldots, Z_{Y_d}) \sim \Gaussp} \max_a
 \Abs{s\left(\Gamma_{1 - \eta}F(Z_X)\right) - s\left(\Gamma_{1 -
 \eta} F(Z_{Y_a})\right)} \geq 
 c \sqrt{ \e \log d} \E_{Z_X,Z_Y \sim \mu_{\Gauss}} \Abs{s\left(\Gamma_{1-\eta}
 F(Z_X)\right) - s\left(\Gamma_{1-\eta}F(Z_Y)\right) }
\end{eqnarray}  
Now we apply the invariance principle on the polynomial $(\Gamma_{1-\eta} F,
\Gamma_{1 - \eta} F)$ and the functional $\Psi : \R^2 \to \R$ given by
$\Psi(a,b) = |s(a)-s(b)|$.  This yields,
\begin{eqnarray}
\E_{Z_X,Z_Y \sim \mu_{\Gauss}} \Abs{s\left(\Gamma_{1-\eta}
F(Z_X)\right) - s\left(\Gamma_{1-\eta}F(Z_Y)\right) } 
& \geq \E_{X,Y \sim \mu(H^R)}
\Abs{s\left(\Gamma_{1-\eta} F(X)\right) -
s\left(\Gamma_{1-\eta} F(Y)\right)} -
\tau^{\Omega(\epsilon\eta/\log(1/\alpha))}
\end{eqnarray}
Over $H^R$, the function $\Gamma_{1-\eta} F$ is bounded in $[0,1]$,
which implies that $s(\Gamma_{1-\eta} F(X)) = \Gamma_{1-\eta} F(X)$
and $\Gamma_{1-\eta} F(X) \geq F(X) - \eta$.
\begin{eqnarray} \label{eq:num6}
\E_{X,Y \sim \mu(H^R)} \Abs{s\left(\Gamma_{1-\eta} F(X)\right) -s\left(\Gamma_{1-\eta} F(Y)\right)}
& \geq \E_{X,Y \sim \mu(H^R)}
\Abs{ F(X) - F(Y)}  - 2\eta
\end{eqnarray}
From equations \prettyref{eq:num1} to \prettyref{eq:num6} we get,
 \begin{eqnarray*} 
\E_{(X,Y_1, \ldots, Y_d) \sim \cP_{H^R} } \max_a  \Abs{F(X) - F(Y_a) }
& \geq \Omega(\sqrt{\epsilon \log d}) \E_{X,Y \sim \mu(H^R)}
\Abs{ F(X) - F(Y)}  - 4\eta - \tau^{\Omega(\epsilon
\eta/\log(1/\alpha))}
\end{eqnarray*}

\end{proof}

\section{Hardness Reduction from \sse}
\label{sec:hardness}

In this section we will present a reduction from \smallsetexpansion
problem to \analyticvsep problem.

Let $G = (V,E)$ be an instance of \smallsetexpansion$(\gamma,\delta,
M)$.  Starting with the instance $G = (V,E)$ of
$\smallsetexpansion(\gamma,\delta,M)$,
our reduction produces an instance $(\cV',\cP')$ of \analyticvsep.  

To describe our reduction, let us fix some notation.  For a set $A$, let $\mset{A}{R}$ denote the set of all
		multisets with $R$ elements from $A$.  
  Let $G_{\eta} = (1-\eta)G + \eta K_{V}$ where $K_{V}$ denotes the complete
graph on the set of vertices $V$.  For an integer $R$, define $G_{\eta}^{\otimes R}$ to be the
product graph $G_{\eta}^{\otimes R} = (G_{\eta})^R$.

Define a Markov chain $H$ on $V_H = \{s,t,t',s'\}$ as follows,$
p(s|s) =
p(s'|s') = 1-\frac{\epsilon}{1-2\epsilon}$, $p(t|s) = p(t'|s') =
\frac{\epsilon}{1-2\epsilon}$, $p(s|t) = p(s'|t') =
\frac{1}{2}$ and $p(t'|t) = p(t|t') = \frac{1}{2}$.  It is easy to see
that the stationary distribution of the Markov chain $H$ over $V_H$ is
given by,
$$ \mu_H(s) = \mu_H(s') = \frac{1}{2} - \epsilon \qquad \qquad \mu_H(t) =
\mu_H(t') = \epsilon$$

%

\noindent The reduction consists of two steps.  First, we construct an
``unfolded'' instance $(\cV, \cP)$ of the \analyticvsep, then we merge
vertices of $(\cV, \cP)$ to create the final output instance $(\cV',
\cP')$.  The details of the reduction are presented below.
%

\begin{mybox}
\textbf{Reduction}

{\sf Input:} A graph $G = (V,E)$ - an instance of
$\smallsetexpansion(\gamma,\delta,M)$.

{\sf Parameters:} $R = \frac{1}{\delta}$, $\epsilon$

{\sf Unfolded instance $(\cV,\cP)$} 

Set $\cV = (V \times V_H)^R$.  The probability distribution $\mu$ on $\cV$
is given by $(\mu_{V} \times \mu_H)^R$.  The probability distribution
$\cP$ is given by the following sampling procedure.

\begin{enumerate}
\item Sample a random vertex $A\in V^R$.
\item Sample $d+1$ random neighbors $B,C_1,\ldots,C_d \sim G_{\eta}^{\tensor R}(A)$ of the vertex
	$A$ in the tensor-product graph $G_{\eta}^{\tensor R}$.
\item Sample $x \in V_H^R$ from the product distribution $\mu^R$.
\item Independently sample $d$ neighbours $y^{(1)},\ldots,y^{(d)} $ of
	$x$ in the Markov chain $H^R$, i.e., $y^{(i)} \sim
	\mu_H^R(x)$.
\item Output $\left( (B,x),(C_1,y_1),\ldots,(C_d,y_d) \right)$
  \end{enumerate}

 {\sf Folded Instance $(\cV',\cP')$}

Fix $\cV' = (V \times \{s,t\})^{\{R\}}$. Define a projection
map $ \Pi: \cV \to \cV'$  as follows:
$$ \Pi(A,x) = \{ (a_i,x_i) | x_i \in \{s,t\} \}$$
for each $(A,x) = \left( (a_1,x_1), (a_2,x_2), \ldots, (a_R,x_R)
\right)$ in $(V \times \{s,t\})^{\{R\}}$.

Let $\mu'$ be the probability distribution on $\cV'$ obtained by projection of
probability distribution $\mu$ on $\cV$.  Similarly, the probability
distribution $\cP'$ on $(\cV')^{d+1}$ by applying the projection $\Pi$ to
the probability distribution $\cP$.

\end{mybox}

Observe that each of the queries $\Pi(B, x)$ and $\{\Pi(C_i,
y_i)\}_{i=1}^d$ are distributed according to $\mu'$ on $\cV'$.
Let $F' \from \cV' \to \{0,1\}$ denote the indicator function
of a subset for the instance.  Let us suppose that 

$$ \E_{X, Y \sim \cV} \left[|F'(X) - F'(Y)| \right] \geq
\frac{1}{10}$$

For the whole reduction, we fix $\eta = \eps/(100 d)$. We will restrict $\gamma < \eps/(100 d)$. 
We will fix its value later. 

\begin{theorem}(Completeness)  
\label{thm:sse-av-completeness}
Suppose there exists a set $S \subset V$
	such that $\vol(S) = \delta$ and $\Phi(S) \leq \gamma$ then
	there exists $F': \cV' \to \{0,1\}$ such that,
$$ \E_{X,Y \sim \cV'} \left[|F'(X) - F'(Y)| \right] \geq
\frac{1}{10} $$	
and,
$$ \E_{X, Y_1,\ldots,Y_d \sim \cP} \left[\max_{i}
	|F'(X) - F'(Y_i)| \right] \leq 2 \epsilon + \bigO{d(\eta+ \gamma)} \leq 4 \e  $$	
\end{theorem}
\begin{proof}
	Define $F : \cV \to \{0,1\}$ as follows: 
	$$ F(A,x) = \begin{cases} 1 & \text{ if } |\Pi(A,x) \cap (S
		\times \{s,t\})| = 1 \\  0 & \text{
		otherwise}\end{cases}$$
	Observe that by definition of $F$, the value of $F(A,x)$ only
	depends on $\Pi(A,x)$.  So the function $F$
	naturally defines a map $F' : \cV' \to \{0,1\}$. 
%
%
%
	Therefore we can write,
	\begin{align*}
		\Pr{F(A,x) = 1} & = \sum_{i \in [R]}  \Pr{x_i \in
			\{s,t\}}\Pr{ \{a_1,\ldots,a_R\}
			\cap S = \{a_i\} | x_i \in \{s,t\} }\\
			& \geq R\cdot \frac{1}{2} \cdot  \frac{1}{R} \cdot
	\left(1-\frac{1}{R}\right)^{R-1}  \geq 	\frac{1}{10}
	\end{align*}
	and,
	$$ \Pr{F(A,x) = 1} = \Pr{ |\Pi(A,x) \cap (S \times
		\{s,t\})| = 1} \leq \E_{(A,x) \sim \cV} \left[ |\Pi(A,x) \cap (S \times
			\{s,t\})| \right] = R \cdot \frac{1}{2} \cdot
			\frac{|S|}{|V|} \leq  \frac{1}{2}$$ 
	
	The above bounds on $\Pr{F(A,x) = 1}$ along with the fact that
	$F$ takes values only in $\{0,1\}$, we get that 
	$$ \E_{X,Y \sim \cV'} \abs{F'(X) - F'(Y)} = \E_{(A,x),(B,y) \sim \cV}{|F(A,x) - F(B,y)|} \geq
	\frac{1}{10} $$
 	
	Suppose we sample $A \in V^R$ and $B,C_1,\ldots,C_d$
	independently from $G_{\eta}^{\otimes R}(A)$.  Let us denote $A =
	(a_1,\ldots,a_R)$, $B = (b_1,\ldots,b_R)$, $C_i =
	(c_{i1},\ldots,c_{iR})$ for all $i \in [d]$.  Note that,
	\begin{eqnarray*}
	\Pr{\exists i \in [R]  \text{ such that } \abs{\{a_i,b_i\} \cap S} =1} 
	& \leq \sum_{i \in [R]}
	(1-\eta)\Pr{(a_i,b_i) \in E[S,\bar{S}]} + \eta \Pr{ (a_i,b_i)
		\in S \times \bar{S}} \\
	&	\leq R (\vol(S) \Phi(S) + 2\eta
		\vol(S)) \leq 	2(\gamma + \eta) \mper
	\end{eqnarray*}
	Similarly, for each $j \in [d]$,
	$$\Pr{ \exists i \in [R] | |\{a_i,c_{ji}\} \cap S|=1} \leq \sum_{i \in [R]}
	\Pr{(a_i,c_{ji}) \in E[S,\bar{S}]} \leq R \vol(S) \Phi(S) \leq
	2(\gamma+ \eta) \mper$$
	By a union bound, with probability at least $1 - 2(d+1)
	(\gamma+ \eta)$
	we have that none of the edges $\{(a_i,b_i)\}_{i\in[R]}$ and
	$\{(a_i,c_{ji})\}_{j\in [d],i \in [R]}$ cross the cut
	$(S,\bar{S})$.

		Conditioned on the above event, we claim that if $(B,x) \cap \left(S \times \{t,t'\}\right)  =
	\emptyset$ then $\max_i |F(B,x) - F(C_i, y_i)|
	=0$.  First, if $(B,x) \cap \left(S \times \{t,t'\}\right)  =
	\emptyset$ then for each $b_i \in S$ the corresponding $x_i
	\in \{s,s'\}$.  In particular, this implies that for each $b_i
	\in S$,  either all of the pairs
	$(b_i,x_i),\{(c_{ji},y_{ji})\}_{j \in [d]}$ are either in $S \times
	\{s,t\}$ or $S \times \{s',t'\}$, thereby ensuring that $\max_i |F(B,x) - F(C_i, y_i)|
	=0$.
	
	From the above discussion we conclude,
	\begin{align*}
	\E_{(B,x), (C_1,y_1),\ldots,(C_d,y_d) \sim \cP} \left[\max_{i}
	|F(B,x) - F(C_i,y_i)| \right] 
	& \leq \Pr{|(B,x) \cap \left(S
		\times \{t,t'\}\right)|\geq 1}  + 2(d +1)(\gamma+\eta) \\	
	& \leq \E\left[|(B,x) \cap \left(S
		\times \{t,t'\}\right)|\right]  + 2(d +1)(\gamma+\eta) \\
	& = R \cdot \vol(S) \cdot \epsilon + 2(d+1) (\gamma+\eta) = \epsilon + 2(d+1) (\gamma+\eta) 
	\end{align*}
\end{proof}

Let $F': \cV' \to \{0,1\}$ be a subset of the instance $(\cV',\cP')$.
Let us define the following notation.
$$ \val_{\cP'}(F') \defeq \E_{(X,Y_1,\ldots,Y_d) \sim \cP'} \left[ \max_{i \in [d]}
\abs{F'(X)-F'(Y_i)}\right]  \qquad \var_1[F'] \defeq \E_{X,Y \sim \cV'}
\abs{F'(X)
- F'(Y)} $$

We define the functions $F : \cV \to [0,1]$ and $f_A, g_A : V_H^R
\to [0,1]$ for each $A \in V^R$ as follows. 
$$ F(A,x) \defeq  F'(\Pi(A,x)) \qquad f_A(x) \defeq F(A,x) \qquad
\qquad g_A(x) \defeq \E_{B \sim G_{\eta}^{\otimes R}(A)} F(B,x)$$

\begin{lemma} \label{lem:val-in-terms-of-g}
	$$ \val_{\cP'}(F') \geq \E_{A \in V^R}
	\val_{\mu_H^R}(g_A)$$
\end{lemma}

\begin{proof}
	\begin{align*}
	\val_{\cP'}(F') & = \val_{\cP}(F)\\
			& = \E_{A \sim V^R} \E_{x \sim \mu_H^R}
			\E_{y_1,\ldots,y_d \sim
			\mu^R_H(x)}\E_{B,C_1,\ldots,C_{d} \sim
			G_{\gamma}^{\otimes
			R}(A)}  \max_i \abs{F(B,x) -
			F(C_i,y_i)} \\
	& \geq \E_{A \sim V^R} \E_{x \sim \mu_H^R} \E_{y_1,\ldots,y_d \sim \mu^R_H(x)} \max_i \abs{ \E_{B \sim
				G_{\gamma}^{\otimes R}(A)} F(B,x) -
			\E_{C_i \sim
				G_{\gamma}^{\otimes R}(A)} F(C_i,y_i)}
				\\
		& \geq \E_{A \sim V^R} \E_{x \sim \mu_H^R}
			\E_{y_1,\ldots,y_d \sim
			\mu^R_H(x)}  \max_i \abs{ g_A(x) -
			g_A(y_i)} \\
			& = \E_{A \in V^R}
	\val_{\mu_H^R}(g_A)			
	\end{align*}
\end{proof}

\begin{lemma} \label{lem:average1}
	$$ \E_{A \sim V^R} \E_{x \sim \mu_H^R} g_A(x)^2 \geq \E_{(A,x)  \sim
	\cV}  F^2(A,x) -  \val_{\cP'}(F') $$
\end{lemma}
\begin{proof}
	\begin{align}
		\E_{A \sim V^R} \E_{x \sim \mu_H^R} g_A(x)^2 & = \E_{A
		\sim V^R} \E_{x \sim \mu_H^R} \E_{B,C \sim G_{\eta}^{\otimes
		R}(A)} F(B,x)F(C,x) \nonumber \\
		&= \frac{1}{2} \E_{A \sim V^R} \E_{x \sim \mu_H^R} \E_{B,C \sim G_{\eta}^{\otimes
		R}(A)} F^2(B,x)+ F^2(C,x)
		-(F(B,x) - F(C,x))^2
		 \nonumber \\
		& = \E_{A \sim V^R} \E_{x \sim \mu_H^R} F^2(A,x) -
		\frac{1}{2} \E_{A \sim V^R} \E_{x \sim \mu_H^R} \E_{B,C \sim G_{\eta}^{\otimes
		R}(A)} (F(B,x) - F(C,x))^2 \label{eq:num10}
	\end{align}
	where in the last step we used the fact that $B,C$ have the
	same distribution as $A \sim V^R$. Since the function $F$ is bounded in $[0,1]$, we have
	\begin{align}
	\E_{A \sim V^R} \E_{x \sim \mu_H^R} \E_{B,C \sim G_{\eta}^{\otimes
		R}(A)} (F(B,x) - F(C,x))^2 
		& \leq
			\E_{A \sim V^R} \E_{x \sim \mu_H^R} \E_{B,C \sim G_{\eta}^{\otimes
			R}(A)} \abs{F(B,x) -
			F(C,x)} \label{eq:num11} 
  	\end{align}
	
	\begin{align}
	 \E_{A \sim V^R}  & \E_{x \sim \mu_H^R}  \E_{B,C \sim
		G_{\eta}^{\otimes R}(A)} \abs{F(B,x) -
			F(C,x)}  \nonumber \\
			& \leq 	\E_{A \sim
			V^R} \E_{x \sim \mu_H^R} \E_{y \sim
			\mu^R_H(x)}\E_{B,C, D \sim G_{\eta}^{\otimes
			R}(A)}  \abs{F(B,x) -
			F(D,y)} + \abs{F(C,x) -
			F(D,y)} \nonumber \\
			& = 2\E_{A \sim V^R} \E_{x \sim \mu_H^R} \E_{y \sim
			\mu^R_H(x)}\E_{B,D \sim G_{\eta}^{\otimes
			R}(A)}  \abs{F(B,x) -
			F(D,y)} \quad  \because{(B,D), (C,D) \text{ have same
			distribution} } \nonumber \\
			& \leq 2\E_{A \sim V^R} \E_{x \sim \mu_H^R}
			\E_{y_1,\ldots,y_d \sim
			\mu^R_H(x)}\E_{B,D_1,\ldots,D_{d} \sim G_{\eta}^{\otimes
			R}(A)}  \max_i \abs{F(B,x) -
			F(D_i,y_i)} \nonumber \\
			& = 2\val_{\cP}(F) = 2\val_{\cP'}(F')
			\label{eq:num12}
	\end{align}
	Equations \prettyref{eq:num10}, \prettyref{eq:num11} and
	\prettyref{eq:num12} yield the desired result.
\end{proof}

\begin{lemma} \label{lem:varianceofg}
	$$ \E_{A \sim V^R} \var_1[g_A] = \E_{A \sim V^R} \E_{x,y \in \mu_H^R} \abs{g_{A}(x) -
	g_{A}(y)} \geq  \frac{1}{2}
	(\var_1[F])^2 - \val_{\cP'}(F')$$		
\end{lemma}

\begin{proof}
	Since the function $g_A$ is bounded in $[0,1]$ we can write
	\begin{align}
	\E_{A \sim V^R} \E_{x,y \in \mu_H^R} \abs{g_{A}(x) -
	g_{A}(y)} 
	& \geq \E_{A \sim V^R} \E_{x,y \in \mu_H^R} \left(g_{A}(x) -
	g_{A}(y)\right)^2 \nonumber \\
	& \geq \E_{A \sim V^R} \E_{x \in \mu_H^R} g^2_{A}(x) -
	\E_{A} \E_{x,y \in \mu_H^R} g_{A}(x)g_{A}(y) \label{eq:num14}
	\end{align}
	In the above expression there are two terms.	From
	\prettyref{lem:average1}, we already know that
	\begin{align}
	 \E_{A \sim V^R} \E_{x \in \mu_H^R} g^2_{A}(x) \geq \E_{(A,x)
	 \sim \cV} F^2(A,x) - \val_{\cP'}(F') \label{eq:num15}
 	\end{align}
	Let us expand out the other term in the expression.
	\begin{align}
		\E_{A} \E_{x,y \in \mu_H^R} g_{A}(x)g_{A}(y)  
		& = \E_{A} \E_{B,C  \sim G_{\eta}^{\otimes
		R}(A)} \E_{x,y \in \mu_H^R}  F'(\Pi(B, x)) F'(\Pi(C,
		y)) \label{eq:num16}
	\end{align}
	Now consider the following graph $\cH$ on $\cV'$ defined by
	the following edge sampling procedure.
\begin{itemize} \itemsep=0ex
	\item Sample $A \in V^R$, and $x,y \in \mu_H^R$.
	\item Sample independently $B \sim G_{\eta}^{\otimes R}(A)$ and $C
		\sim G_{\eta}^{\otimes R}(A)$
	\item Output the edge $\Pi(B,x)$ and $\Pi(C,y)$
	\end{itemize}
	Let $\lambda$ denote the second eigenvalue of the adjacency
	matrix of the graph $\cH$.  
	\begin{align*}
\E_{A} \E_{B,C  \sim G_{\eta}^{\otimes
		R}(A)} &\E_{x,y \in \mu_H^R}  F'(\Pi(B, x)) F'(\Pi(C,
		y))  = \iprod{F' , \cH F'}\\ 
		& \leq \left(\E_{(A,x) \sim \cV} F'(\Pi(A,x))\right)^2 + \lambda
		\left(\E_{(A,x) \sim \cV} \left(F'(\Pi(A,x))\right)^2 - (\E_{(A,x) \sim \cV} F'(\Pi(A,x)))^2 \right) \\
		& = \lambda \E_{(A,x) \sim \cV} F(A,x)^2 + (1-\lambda)
		(\E_{(A,x) \sim \cV} F(A,x))^2  
		 \quad \because{$F'(\Pi(A,x)) = F(A,x)$}
	\end{align*}
	Using the above inequality with equations
	\prettyref{eq:num14}, \prettyref{eq:num15},
	\prettyref{eq:num16} we can derive the following,
	\begin{align*}
	\E_{A \sim V^R} \E_{x,y \in \mu_H^R} \abs{g_{A}(x) -
	g_{A}(y)} 
	& \geq \E_{A \sim V^R} \E_{x \in \mu_H^R} g^2_{A}(x) -
	\E_{A} \E_{x,y \in \mu_H^R} g_{A}(x)g_{A}(y) \\
	& \geq (1-\lambda) \left[ \E_{(A,x) \sim \cV} F^2(A,x) -
		(\E_{(A,x) \sim \cV} F(A,x))^2 \right] -
		\val_{\cP'}(F')\\
		& \geq (1-\lambda) \var[F] - \val_{\cP'}(F') \\
		&\geq (1-\lambda) (\var_1[F])^2	- \val_{\cP'}(F')
		\quad \because{$\var[F] > \var_1[F]^2$ for all $F$} 	
	\end{align*}
	To finish the argument, we need to bound the second eigenvalue
	$\lambda$ for the graph $\cH$. 	Here we will present a simple
	argument showing that the second eigenvalue $\lambda$ for
	the graph $\cH$ is strictly less than $\frac{1}{2}$. 
	Let us restate the procedure to sample edges from $\cH$
	slightly differently.
	\begin{itemize} \itemsep=0ex
	\item Define a map $\cM : V \times V_H \to (V \cup \perp) \times
		(V_H \cup \{\perp\})$ as follows, $\cM(b,x) = (b,x)$
		if $x \in \{s,t\}$ and $\cM(b,x) = (\perp,\perp)$
		otherwise.  Let $\Pi' :  ((V \cup \perp) \times
		(V_H \cup \perp) )^R \to (V \times \{s,t\})^{\{R\}}$ denote the
		following map.  $$\Pi'(B',x') = \{(b'_i,x'_i) | x_i
			\in \{s,t\}\}$$
 	\item Sample $A \in V^R$ and $x,y \in \mu_H^R$
	\item Sample independently $B = (b_1,\ldots,b_R) \sim G_{\eta}^{\otimes R}(A)$
		and $C = (c_1,\ldots,c_R) \sim G_{\eta}^{\otimes R}(A)$.
	\item Let $\cM(B,x),\cM(C,y)  \in \left((V \cup \{\perp\}) \times
		(V_H \cup \{\perp\}) \right)^R$ be obtained by applying $\cM$ to each
		coordinate of $(B,x)$ and $(C,y)$.
	\item 	Output an edge between $(\Pi'(\cM(B,x)),
		\Pi'(\cM(C,y)))$.	
	\end{itemize}
	It is easy to see that the above procedure also samples the
	edges of $\cH$ from the same distribution as earlier.  Note
	that $\Pi'$ is a projection from $((V \cup \perp) \times
		(V_H \cup \perp))^R$  to $(V \times \{s,t\})^{\{R\}}$.  Therefore,
		the second eigenvalue of the graph $\cH$ is upper
		bounded by the second eigenvalue of the graph
		on $((V \cup \perp) \times
		(V_H \cup \{\perp\}))^R$ defined by $\cM(B,x) \sim \cM(C,y)$.
	Let $\cH_1$ denote the graph defined by the edges $\cM(B,x) \sim \cM(C,y)$.
	Observe that the coordinates of $\cH_1$ are independent, i.e.,
	$\cH_1 = \cH_2^R$ for a graph  $\cH_2$ corresponding to each coordinate of $\cM(B,x)$
	and $\cM(C,y)$.  Therefore, the second eigenvalue of $\cH_1$
	is at most the second eigenvalue of $\cH_2$.  The Markov chain $\cH_2$ on $ (V \cup \{\perp\}) \times (V_H
	\cup \perp)$
	is defined as follows,
	\begin{itemize}
	\itemsep=0ex
	\item Sample  $a \in V$ and two neighbors $b \sim G_{\eta}(a)$
		and $c \sim G_{\eta}(a)$.
	\item Sample $x,y \in V_H$ independently from the distribution $\mu_H$.
	\item Output an edge between $\cM(b,x)$ $\cM(c,y)$. 
	\end{itemize}
	Notice that in the Markov chain $\cH_2$, for every choice of
	$\cM(b,x)$ in $(V
	\cup \{\perp\}) \times (V_H \cup \perp)$, with probability at least
	$\frac{1}{2}$, the other endpoint $\cM(c,y) = (\perp,
	\perp)$.  Therefore, the second eigenvalue of $\cH_2$ is at
	most $\frac{1}{2}$, giving a bound of $\frac{1}{2}$ on the
	second eigen value of $\cH$.
\end{proof}

Now we restate a claim from \cite{rst12} that will be useful
for our our soundness proof.
\begin{theorem} \label{thm:sse-ug-soundness} (Restatment of Lemma 6.11
	from \cite{rst12})
Let $G$ be a graph with a vertex set $V$.  Let a distribution on pairs
of tuples $(A,B)$ be defined by $A \sim V^R$, $B \sim G_{\eta}^{\otimes
R}(A)$.  Let $\ell: V^R \to [R]$ be a labelling such that over the
choice of random tuples and two random permutations $\pi_A, \pi_B$
$$ \pr_{A \sim V^R, B \sim G_{\eta}^{\otimes R}(A)} \pr_{\pi_A,\pi_B} \left\{
 \pi_A^{-1}\left(\ell(\pi_A(A))\right) =
 \pi_B^{-1}\left(\ell(\pi_B(B))\right) \right\} \geq \zeta$$
 Then there exists a set $S \subset V$ with $\vol(S) \in
 \left[\frac{\zeta}{16R}, \frac{3}{\eta R} \right]$ satisfying
 $\Phi(S) \leq 1- \zeta/16$.
\end{theorem}

The following lemma asserts that if the graph $G$ is a $NO$-instance
of \smallsetexpansion($\gamma$, $\delta$,$M$) then for almost all $A
\in V^R$ the functions have no influential coordinates. 
\begin{lemma} \label{lem:ugdecoding}
	Fix $\delta = 1/R$.  Suppose for all sets $S \subset V$ with $\vol(S) \in \left(\delta/M,
	M\delta\right) $ , $\Phi(S) \geq 1-\gamma$ then for all
	$\tau > 0$, 
	$$ \pr_{A \sim V^R}\left[ \exists i \mid \Inf_{i}[\Gamma_{1-\eta} g_A] \geq
	\tau \right] \leq \frac{1000}{\tau^3 \eps^2 \eta^2} \cdot \max(1/M,
	\gamma)  $$
\end{lemma}

\begin{proof}
For each $A \in V^R$, let $L_A = \set{ i\in
  [R]\mid \Inf_i (\Gamma_{1-\eta}  f_A)>\tau/2}$ and $L'_A = \set{
	  i\in [R]\mid \Inf_i (\Gamma_{1-\eta} g_A)>\tau}$. 
Call a vertex $A \in V^R$ to be {\it good} if $L'_A \neq \emptyset$.
By \prettyref{lem:sum-of-influences1}, the sum of influences of $\Gamma_{1-\eta} g_A$ is at most
$\frac{1}{\epsilon \eta} \var[g_A] \leq \frac{1}{\epsilon \eta}$.
Therefore, the cardinality of $L'_A$ is upper bounded by $|L'_A| \leq
\frac{2}{\tau \epsilon \eta}$.  Similarly, the cardinality of $L_A$ is
upper bounded by $|L_A| \leq \frac{1}{\tau \epsilon \eta}$.

The lemma asserts that at most a $\frac{1000}{\tau^3\eta^2 \epsilon^2 } \cdot
\max(1/M,\gamma)$ fraction
of vertices are {\it good}. For the sake of contradiction, assume that $\pr_{A \in V^R}
	  \left[ L'_A \neq \emptyset \right] \geq 1000 \max(1/M,\gamma) 
	  /\tau^2 \epsilon^2 \eta^2$.

Define a labelling $\ell : V^R \to [R]$ as follows:  for each $A \in
V^R$, with probability $\frac{1}{2}$ choose a random coordinate in
$L_A$ and
with probability $\half$, choose a random coordinate in $L'_A$.
 If the sets $L_A,L'_A$ are empty, then we choose a
uniformly random coordinate in $[R]$.


Observe that for each $A\in V^R$, the function $g_A$ is the average over bounded
functions $f_B\from V_H^R\to[0,1]$, where $B \sim G^R_{\eta}(A)$.
Fix a vertex $A \in V^R$ such that $L'_A \neq \emptyset$ and a
coordinate $i \in L'_A$.  In particular, we have that
$\Inf_{i}[\Gamma_{1-\eta} g_A] \geq \tau$.  Using convexity of influences,
this implies that,
$$E_{B \sim G^{\otimes R}_\eta(A)} \Inf_i[\Gamma_{1-\eta} f_B] \geq
	\tau \mper$$
Specifically, this implies that for at least a $\frac{\tau}{2}$ fraction of the neighbours $B \sim
G^R_{\eta}(A)$, the influence of the $i^{th}$ coordinate on $f_B$ is at
least $\frac{\tau}{2}$.  Hence, if $L'_A \neq \emptyset$ then for at least a $\tau/2$ fraction
of neighbours $B \sim G^{\otimes R}_\eta(A)$ we have $L'_A \cap L_B
\neq \emptyset$.

By definition of the functions $f_A, g_A$, it is clear that for every
permutation $\pi : [R] \to [R]$, $f_A(\pi(x)) = f_{\pi(A)}(x)$ and $g_A(\pi(x)) =
g_{\pi(A)}(x)$. Therefore, for every permutation $\pi : [R] \to [R]$
and $A \in V^R$,
$$ L_A =  \pi^{-1}(L_{\pi(A)}) \qquad \text{ and } L'_A =
\pi^{-1}(L'_{\pi(A)})$$

	  From the above discussion, for every {\it good} vertex $A \in V^R$, for at least a $\tau/2$
fraction of the vertices $B \sim G^{\otimes R}_{\eta}(A)$, and every pair
of permutations $\pi_A,\pi_B : [R] \to [R]$, we have  $\pi^{-1}_A(L'_{\pi_A(A)})
\cap \pi^{-1}_B(L_{\pi_B(B)}) \neq \emptyset$.  This implies that,

\begin{align*}
 & \pr_{A \sim V^R, B \sim G_{\eta}^{\otimes R}(A)} \pr_{\pi_A,\pi_B} \left\{
	 \pi_A^{-1}\left(\ell(\pi_A(A))\right) =
 \pi_B^{-1}\left(\ell(\pi_B(B))\right) \right\} \\
 & \geq \pr_{A \sim V^R} [L'_A \neq \emptyset] \cdot \pr_{B \sim
	 G^{\otimes R}_{\eta}(A)}[L'_A \cap L_B \neq \emptyset | L'_A
	 \neq \emptyset] \cdot 
 \Pr{\pi_{A}^{-1}(\ell(\pi_{A}(A))) =
		\pi_{B}^{-1}(\ell(\pi_{B}(B))) \mid L'_A \cap L_B \neq
	\emptyset} \\
& \geq 	\pr_{A \sim V^R} [L'_A \neq \emptyset] \cdot \left(\frac{\tau}{2}\right) \cdot \frac{1}{2} \cdot \frac{1}{2} \cdot
\frac{1}{|L'_A|} \frac{1}{|L_B|} \\
& \geq 	\pr_{A \sim V^R} [L'_A \neq \emptyset] \cdot \left(\frac{\tau}{2}\right) \cdot \frac{1}{2} \cdot \frac{1}{2} \cdot
\left(\frac{\tau \eta \epsilon}{2}\right)^2 \\
& \geq 16\max(\nfrac{1}{M},\gamma)
\end{align*}
By \prettyref{thm:sse-ug-soundness}, this implies that there exists a
set $S \subset V$ with $\vol(S) \in [ \frac{1}{MR}, \frac{3}{\eta R}]$
satisfying $\Phi(S) \leq 1-\gamma$.  A contradiction.
\end{proof}
Finally, we are ready to show the soundness of the reduction.  
\begin{theorem} (Soundness)
\label{thm:sse-av-sound}
For all $\epsilon,  d$ there exists choice of $M$ and $\gamma, \eta$
such that the following holds. 
Suppose for all sets $S \subset V$ with $\vol(S) \in \left(\delta/M,
	M\delta\right) $ , $\Phi(S) \geq 1-\eta$, then for all $F' :
	\cV' \to [0,1]$ such that $ \var_1[F'] \geq \frac{1}{10}$, we
	have $\val_{\cP'}(F') \geq \Omega(\sqrt{\epsilon \log{d}})$
\end{theorem}
\begin{proof}
	Recall that we had fixed $\eta = \eps/(100 d)$.  We will choose $\tau$ to small enough
	so that the error term in the soundness of dictatorship test
	(\prettyref{prop:soundgadget}) is smaller than $\epsilon$.
	Since the least probability of any vertex in Markov
	chain $H$ is $\epsilon$, setting $\tau =
	\epsilon^{1/\epsilon^3}$ would suffice.

	First, we know that if $G$ is a $NO$-instance of
	\smallsetexpansion($\gamma,\delta,M$) then for almost all $A
	\in V^R$, the function $g_A$ has no influential coordinates.
	Formally, by \prettyref{lem:ugdecoding}, we will have  	$$ \pr_{A \sim
	V^R}\left[ \exists i \mid \Inf_{i}[\Gamma_{1-\eta} g_A] \geq
	\tau \right] \leq \frac{1000}{\tau^3\eta^2} \cdot \max(1/M,
	\gamma)  \mper $$
	For an appropriate choice of $M, \gamma$, the above inequality
	implies that for all but an $\epsilon$-fraction of vertices $A
	\in V^R$, the function $g_A$ will have no influential
	coordinates.

	Without loss of generality, we may assume that
	$\val_{\cP'}(F') \leq \sqrt{\epsilon \log d}$, else we would
	be done.  Applying \prettyref{lem:varianceofg}, we get that $\E_{A \in V^R}
	\var_1[g_A] \geq (\var_1[F])^2 - \val_{\cP'}(F') \geq
	\frac{1}{200}$.  This implies
	that for at least a
	$\frac{1}{400}$ fraction of $A \in V^R$, $\var_1[g_A] \geq
	1/400$.  Hence for at least an $1/400-\epsilon$ fraction of vertices $A
	\in V^R$ we have,
	$$ \var_1[g_A] \geq \frac{1}{400} \qquad \text{ and } \qquad
	\max_i \Inf_i(\Gamma_{1-\eta}(g_A)) \leq \tau$$
	By appealing to the soundness of the gadget
	(\prettyref{prop:soundgadget}), for every such vertex $A \in
	V^R$,
	$\val_{\mu_H^R}(g_A) \geq \Omega(\sqrt{\epsilon \log{d}}) -
	O(\epsilon) = \Omega(\sqrt{\epsilon \log{d}})$.
	Finally, by applying \prettyref{lem:val-in-terms-of-g}, we get
	the desired conclusion.
	$$ \val_{\cP'}(F') \geq \E_{A \in V^R} \val_{\mu_H^R}(g_A)
	\geq \Omega(\sqrt{\epsilon \log{d}})$$
\end{proof}

\section{Reduction from {\em Analytic $d$-Vertex Expansion} to Vertex Expansion }
\label{sec:sampling}


\begin{theorem}
\label{thm:reduction}
A \cvss hardness for $d$-\analyticvsep implies a \cvs{4}{16} hardness for 
balanced symmetric-vertex expansion on graphs of degree at most $D$, where
$D = \max \set{100 d/s, 2 \log (1/c)}$.

\end{theorem}

At a high level, the proof of \prettyref{thm:reduction} has two steps. 
\begin{enumerate}
\item 
We show that a \cvss hardness for \analyticvsep. 
implies a \cvs{2}{4} hardness for instances of \analyticvsep having uniform distribution
(\prettyref{prop:uniform}).

\item We show that a \cvss hardness for instances of $d$-\analyticvsep having  uniform stationary distribution
implies a \cvs{2}{2} hardness for balanced symmetric-vertex expansion on $\Theta(D)$-regular graphs.
(\prettyref{prop:subsample}).

\end{enumerate}

\begin{proposition}
\label{prop:uniform}
A \cvss hardness for \analyticvsep. 
implies a \cvs{2}{4} hardness for instances of \analyticvsep having uniform distribution.  
\end{proposition}

\begin{proof}

Let $(V,\cP)$ be an instance of \analyticvsep.
We construct an instance $(V',\cP')$ as follows. Let $T = 2 n^2 $.
We first delete all vertices $i$ from $V$ which have $\mu(i) < 1/2n^2$, i.e. 
$V \gets V \backslash \set{i \in V : \mu(i) < 1/2n^2}$.
Note that after this operation, the total weight of the remaining vertices is still at least 
$1 - 1/2n$ and the \analyticvsep  can increase or decrease 
by at most a factor of $2$.
Next  for each $i$, we introduce introduce $\lceil \mu(i) T \rceil$ copies of vertex $i$. We will call these vertices the cloud for vertex $i$ 
and index them as  $(i,a)$ for $a \in [\mu(i) T]$.

We set the probability mass of each $(d+1)$-tuple $( (i,a), (j_1,b_1) \ldots, (j_d,b_d))$    as follows : 
\[ \cP'( (i,a), (j_1,b_1) \ldots, (j_d,b_d) )  =  \frac{ \cP(i, j_1, \ldots, j_d)}{( \mu(i) T) \cdot \Pi_{\ell = 1}^d (\mu(j_{\ell}) T)    }  \]


It is easy to see that $\mu'(i,a) = 1/T$ for all vertices $(i,a) \in V'$.
 The analytic $d$-vertex expansion under a function $F$ is given by,
\[ \frac{\E_{( (i,a), (j_1,b_1) \ldots, (j_d,b_d) ) \sim \cP'} \max_{\ell} \Abs{F(i,a) - F(j_{\ell},b_{\ell})}}
	{\E_{ (i,a),(j,b) \sim \mu'} \Abs{ F(i,a) - F(j,b)} }  \]

where $X = (i,a)$ and  $Y_{\ell} = (j, b)$  which are sampled as follows:

\begin{enumerate}
\item Sample a $(d+1)$-tuple $(i,j_1,\ldots,j_d)$ from $\cP$.

\item Sample $a$ uniformly at random from ${1, \ldots,  \mu(i) T}$.

\item Sample $b_{\ell}$ uniformly at random from $ \set{1, \ldots, \mu(j_{\ell}) T}$ for each $\ell \in [d]$.
\end{enumerate}

\paragraph{Completeness}
Suppose, $\phiav{V,\cP} \leq c$. Let $f$ be the corresponding cut function.  
The function $f : V \to \set{0,1}$ 
can be trivially extended to a function $F: V' \to \set{0,1}$ 
thereby certifying that $\phiav{V',\cP'} \leq 2 c$.

\paragraph{Soundness}
Suppose $\phiav{V,\cP} \geq s$. Let $F : V' \to \set{0,1}$ be any balanced function. 
By convexity of absolute value function we get

\[ \E_{ ((i,a),(j_1,b_1), \ldots, (j_d,b_d)) \sim \cP'} \max_{\ell} \Abs{F(i,a) - F(j_{\ell,b_{\ell}})}  \geq 
  \E_{(i,j_1, \ldots, j_d) \sim \cP} \max_\ell \Abs{ \E_a F(i,a) - \E_{\ell} F(j_\ell, b_\ell)}. \]

So if we define $f(i) = E_a F(i,a)$, the numerator for analytic $d$-vertex expansion in $(V,\cP)$ for $f$
is only lower than the corresponding numerator for $F$ in $(V',\cP')$.
We need to lower bound the denominator, $ \E_{i,j \sim \mu}  \Abs{f(i) - f(j)}$.
The requisite lower bound follows from the following two lemmas.

\begin{lemma}
\label{lem:help1}
\[  \E_{i,j \sim \mu}  \Abs{f(i) - f(j)} \geq \E_{(i,a),(j,b) \sim \mu'} \Abs{F(i,a) - F(j,b)} - \E_{(i,a), (i,b) \sim \mu'} \Abs{F(i,a) - F(i,b)}  \] 

\end{lemma}

\begin{proof}   
The Lemma follows directly from the following two inequalities.
\[  \E_{(i,a), (j,b)} \Abs{ F(i, a) - F(j,b)}   \leq  \E_{(i,a)} \Abs{ F(i,a) - f(i)} + \E_{(j,b)} \Abs{ F(j,b) - f(j)}  + \E_{i,j} \Abs{f(i) - f(j)}  
	\qquad \textrm{(Triangle Inequality)}
 \]
and
\[ \E_{i,a} \Abs{ F(i,a) - f(i) } \leq \E_{i, a,b} \Abs{ F(i,a) - F(i,b) } \]
\end{proof}

\begin{lemma}
\label{lem:help2}
\[  \E_{i,a,b} \Abs{F(i,a) - F(i,b)}   \leq  2\val_{\cP'}(F) =
	  2 \E_{(i,a),(j_1,c_1), \ldots (j_d,c_d) \sim \cP' } \max_{\ell} \Abs{F(i,a) - F(j_\ell ,c_\ell)}   \]
\end{lemma}

\begin{proof}
Sample $(i,j_1, \ldots, j_d) \sim \cP$. 
For any neighbour $(j,c)$  of $(i,a), (i,b)$, using the Triangle Inequality we have  
\[ \Abs{F(i,a) - F(i,b)}   \leq \Abs{F(i,a) - F(j,c) }  +\Abs{F(j,c) - F(i,b) } \]

Therefore,
\begin{eqnarray*}
\Abs{F(i,a) - F(i,b)} & \leq & 
\frac{\sum_{\ell} \Abs{F(i,a) - F(j_\ell ,c_\ell)}  +  \sum_{\ell} \Abs{F(i,b) - F(j_\ell ,c_\ell)} }{d} \\
& \leq & \max_{\ell} \Abs{F(i,a) - F(j_\ell ,c_\ell)}  +  \max_{\ell} \Abs{F(i,b) - F(j_\ell ,c_\ell)} \\
\end{eqnarray*}

Taking expectations over the uniformly random choice of $a$ and $b$ from the cloud of $i$,
\[ \E_{(i,a), (i,b)}  \Abs{ F(i,a) - F(i,b)}  \leq 2 \E_{ ((i,a),(j_1,b_1), \ldots, (j_d,b_d)) \sim \cP'}  \max_{\ell} \Abs{F(i,a) - F(j_\ell ,c_\ell)} \]
\end{proof}

\prettyref{lem:help1} and \prettyref{lem:help2} together show that 
\[ \E_{i,j} \Abs{f(i) - f(j)} \geq  \frac{ \E_{(i,a), (j,b)} \Abs{ F(i,a) - F(j,b)} }{2}.  \]
as long as the value $\val_{\cP'}(F) < \var_1[F]/4$.
Therefore, for any $F : V' \to \set{0,1}$, 
\[ \frac{\E_{( (i,a), (j_1,b_1) \ldots, (j_d,b_d) ) \sim \cP'} \max_{\ell} \Abs{F(i,a) - F(j_{\ell},b_{\ell})}}
	{\E_{ (i,a),(j,b) \sim \mu'} \Abs{ F(i,a) - F(j,b)} }  \geq \frac{s}{4} \mper \]

\prettyref{thm:cont-to-bin} shows that the minimum value of \analyticvsep is obtained by boolean functions.
Therefore, $\phiav{V',\cP'} \geq s/4$.
\end{proof}

\begin{proposition}
\label{prop:subsample}
A \cvss hardness for instances of $d$-\analyticvsep having  uniform stationary distribution
implies a \cvs{2}{4} hardness for balanced symmetric-vertex expansion on 
$\Theta(D)$-regular graphs. Here $D \geq \max \set{100 d/s, 2 \log (1/c)} $.

\end{proposition}

\begin{proof}

Let $(V',\cP')$ be an instance of $d$-\analyticvsep. 
We construct a graph $G$ from $(V',\cP')$ as follows. We initially set $V(G) = V'$. For each vertex $X$ we pick $D$ neighbors
by sampling $D/d$ tuples from the marginal distribution of $\cP'$ on tuples containing $X$ in the first coordinate.

Let $\deg(i)$ denote the degree of vertex $i$, i.e. the number of vertices adjacent to vertex $i$ in $G$. 
It is easy to see that $\deg(i) \geq D$ and $\Ex{\deg(i)} = 2D\ \forall i \in V(G)$. 
Let $L = \set{i \in V(G) | \deg(i) > 4 D}$. Using Hoeffding's Inequality,
we get a tight concentration for $\deg(i)$ around $2D$.
\[  \Pr{ \deg(i) >  4D} \leq e^{-D} \mper \]
Therefore, $\Ex{ \Abs{L} } < n/e^D$. We delete these vertices from $G$, i.e. $V(G) \gets V(G) \backslash L$. 
With constant probability, all remaining vertices will have their degrees in the range $[D/2, 4D]$. 
Also, the vertex expansion of every set will decrease by at most an additive $1/e^D$.

\paragraph{Completeness}
Let $\phiav{V',\cP'} \leq c$ and
let $F : V' \to \set{0,1}$ be the function corresponding to $\phiav{V',\cP'}$. 
Let the set $S$ be the support of the function $F$. 
Clearly, the set $S$ is balanced.
Therefore, with constant probability, we have 
\[ \phivs(G) \leq  \phivs_G(S) \leq  \phiav{V',\cP'} + 1/e^D \leq 2c \mper \] 

\paragraph{Soundness}
Suppose $\phiav{V',\cP'} \geq s$.
Let $F : V' \to \set{0,1}$ be any balanced function.

Since the max is larger than the average, we get  
\[ \E_{X} \max_{Y_i \in N_G(X) } \Abs{F(X) - F(Y_i)} \geq \frac{d}{D} \sum_{j = 1}^{D/d} \E_{ (X,Y_1, \ldots, Y_d) \sim \cP} \max_i \Abs{F(X) - F(Y_i)} \]

By Hoeffding's inequality, we get 
\begin{eqnarray*}
\Pr{ \left( \E_{X} \max_{Y_i \in N(X) } \Abs{F(X) - F(Y_i)} \right) < s/4} & \leq &
\Pr{ \left( \frac{d}{D} \sum_{j = 1}^{D/d} \E_{ (X,Y_1, \ldots, Y_d) \sim \cP} \max_i \Abs{F(X) - F(Y_i)} \right) < s/4  } \\
& \leq & \exp \left(  - n  (s D/d)^2 \right) 
\end{eqnarray*}
Here, the last inequality follows from Hoeffding's inequality over the index $X$.
There are at most $2^n$ boolean functions on $V$. Therefore, using a union bound on all those functions we get,
\[ \Pr{\phivs(G) \geq s/4}  \geq 1 - 2^n \exp \left(  - n (sD/d)^2 \right). \]

Since $D > d/s$, 
we get that with probability $1 - o(1)$, $\phivs(G) \geq s/4$.

\end{proof}

\begin{proof}[Proof of \prettyref{thm:reduction}]
\prettyref{thm:reduction} follows directly from
\prettyref{prop:uniform} and \prettyref{prop:subsample}.
\end{proof}

\section{Hardness of Vertex Expansion}
\label{sec:putting-things-togethor}

We are now ready to prove \prettyref{thm:main}. We restate the Theorem below.

\begin{theorem}
For every $\eta > 0$, there exists an absolute constant $C$ such that $\forall \e>0 $ it is \sse-hard to distinguish 
between the following two cases for a given graph $G = (V,E)$ with maximum degree $d \geq 100/\e$.
\begin{description}
\item[\yes] : There exists a set $S \subset V$ of size $\Abs{S} \leq \Abs{V}/2$ such that 
	\[ \phiv(S) \leq \e  \]
\item[\no] : For all sets $S \subset V$, 
	\[  \phiv(S) \geq  \min \set{10^{-10}, C \sqrt{\e \log d}} - \eta \]
\end{description}
\end{theorem}

\begin{proof}

From \prettyref{thm:sse-av-completeness} and \prettyref{thm:sse-av-sound}
we get that for an instance of \analyticvsep $(V,\cP)$,
 it is \sse-hard to distinguish between the following two cases cases: 

\begin{description}
\item[\yes] : 	\[ \phiav{V,\cP} \leq \e \]

\item[\no] : \[  \phiav{V,\cP} \geq  \min \set{ 10^{-4}, c_1 \sqrt{\e \log d} } - \eta \]

\end{description}

Now from \prettyref{thm:reduction} we get that for a graph $G$,
it is \sse-hard to distinguish between the following two cases cases:

\begin{description}
\item[\yes] : 	\[ \phivbs \leq \e  \]

\item[\no] : \[  \phivbs \geq  \min \set{ 10^{-6}, c_2 \sqrt{\e \log d} } - \eta \]

\end{description}

We use a standard reduction from Balanced vertex expansion to vertex expansion. 
For the sake of completeness we give a proof of this reduction in
\prettyref{lem:bal-vert-to-vert}.
Using this reduction, we get that for a graph $G$,
it is \sse-hard to distinguish between the following two cases cases:

\begin{description}
\item[\yes] : 	\[ \phivs \leq \e  \]

\item[\no] : \[  \phivs \geq  \min \set{10^{-8}, c_3 \sqrt{\e \log d} } - \eta \]

\end{description}

Finally, using the computational equivalence of vertex expansion and symmetric vertex expansion (\prettyref{thm:symv}), 
we get that for a graph $G$,
it is \sse-hard to distinguish between the following two cases cases:

\begin{description}
\item[\yes] : 	\[ \phiv \leq \e  \]

\item[\no] : \[  \phiv \geq  \min \set{ 10^{-10}, C \sqrt{\e \log d} } - \eta \]

\end{description}

This completes the proof of the theorem.

\end{proof}

\section{An {\em Optimal} Algorithm for vertex expansion }
\label{sec:vertexsepalgo}

In this section we give a simple polynomial time algorithm which outputs a set $S$
whose vertex expansion is at most $\bigO{\sqrt{ \phiv \log d}}$. We
restate \prettyref{thm:algo}.

\begin{theorem}

There exists a polynomial time algorithm which given a graph $G = (V,E)$ having vertex degrees
at most $d$, outputs a set $S \subset V$, such that $ \phiv(S) = \bigO{\sqrt{ \phiv \log d}}$.

\end{theorem}

For an undirected graph $G$,
Bobkov \etal \cite{bht00} define $\linf$ as follows. 
\[ \linf \defeq \min_{ x} \frac{ \sum_i \max_{j \sim i} (x_i - x_j)^2 }{ \sum_i x_i^2 - \frac{1}{n}(\sum_i x_i)^2  }  \]

They also prove the following Theorem.

\begin{theorem}[\cite{bht00}]
\label{thm:bht}
For any unweighted, undirected graph $G$, we have 
\[  \frac{\linf}{2}  \leq \fv \leq \sqrt{2 \linf}  \]

\end{theorem}

Consider the following \sdp relaxation of $\linf$.

\begin{mybox}

\begin{SDP}
\label{sdp:linf}
\begin{eqnarray*}
\sdpval \defeq \min \sum_{i \in } \alpha_i && \\
\textrm{subject to:} \qquad \qquad  && \\ 
\norm{v_j - v_i }^2 & \leq & \alpha_i \qquad \forall i \in V \textrm{ and } \forall j \sim i \\
\sum_i \norm{v_i}^2 - \frac{1}{n}\norm{\sum_i v_i}^2 & = & 1 \\
\end{eqnarray*}

\end{SDP}

\end{mybox}

It's easy to see that this is a relaxation for $\linf$.
We present a simple randomized rounding of this \sdp which, with constant probability,
outputs a set with vertex expansion at most $C \sqrt{\phiv \log d}$ for some absolute constant $C$.


\begin{mybox}

\begin{algorithm}

\label{alg1}

\begin{itemize}~

\item {\em Input : } A graph $G = (V,E)$ 
\item {\em Output :} A set $S$ with vertex expansion at most $ 576 \sqrt{\sdpval \log d}$ (with constant probability).

\begin{enumerate}

\item Compute graph $G'$ as in \prettyref{thm:symv2}, let $n = \Abs{V(G')}$.

\item Solve \prettyref{sdp:linf} for graph $G'$.

\item Pick a random Gaussian vector $g \sim N(0,1)^n$. 

\item For each $i \in [n]$, define $x_i \defeq \inprod{v_i,g}$.

\item Sort the $x_i$'s in decreasing order $x_{i_1} \geq x_{i_2} \geq \ldots x_{i_n}$.
Let $S_j$ denote the set of the first $j$ vertices appearing in the sorted order. 
Let $l$ be the index such that
\[  l = {\sf argmin}_{1 \leq j \leq n/2} \phivs(S_j) \mper \]

\item Output the set corresponding to $S_l$ in $G$.

\end{enumerate}

\end{itemize}

\end{algorithm}

\end{mybox}



We first prove a technical lemma which shows that we can a recover a
a set with small vertex expansion from a { \em good} linear-ordering
(Step $3$ in Algorithm \ref{alg1}).

\begin{lemma}
\label{lem:levelset}

For any $y_1, y_2, \ldots, y_n \in \R^+ \cup \set{0}$, let $Y \defeq [y_1 y_2 \ldots y_n]^T$ and
 $\alpha \defeq \frac{ \sum_i \max_{j \sim i} |y_j - y_i|  }{\sum_i y_i}$. Then $\exists S \subseteq \supp(Y)$
such that $\fv(S) \leq \alpha$. Morover, such a set can be computed in polynomial time.  
\end{lemma}

\begin{proof}
W.l.o.g we may assume that $y_1 \geq y_2 \geq \ldots \geq y_n \geq 0$. Then
\begin{equation}
  \frac{ \sum_i \max_{ j \sim i , j < i } (y_j - y_i)  }{\sum_i y_i} \leq \alpha  \label{hin}
\end{equation}
  
{ and } 
\begin{equation}
  \frac{ \sum_i \max_{ j \sim i , j > i } (y_i - y_j)  }{\sum_i y_i} \leq \alpha \label{hout}
\end{equation}

Let $i_{\max} \defeq \textrm{\sf argmax}_i y_i>0 $, i.e. $i_{\max}$ be the largest index such that $y_{i_{\max}} > 0$.
Let $S_i \defeq \set{y_1, \ldots, y_i}$.  
Suppose $\forall i < i_{\max}$ $N^v(S_i) > \alpha |S_i|$.

Now, from Inequality \ref{hout},

\[ \alpha \geq  \frac{ \sum_i \max_{ j \sim i , j < i } (y_j - y_i) }{\sum_i y_i} = 
 \frac{ \sum_i \max_{ j \sim i , j < i }  \sum_{l=j}^{l = i-1} (y_l - y_{l+1})  }{\sum_i y_i}  =  
 \frac{ \sum_i ( y_i - y_{i+1}) |N(S_i)|  }{\sum_i y_i} > \alpha  \frac{ \sum_i ( y_i - y_{i+1}) |S_i|  }{\sum_i y_i} 
= \alpha    \]

Thus we get $\alpha > \alpha$ which is a contradition. Therefore, $\exists i \leq i_{\max}$ such that $\fv(S_i) \leq \alpha$.
\end{proof}

Next we show a $\linf$-like bound for the $x_i$'s. 

\begin{lemma}
\label{lem:projlinf}
Let $x_1, \ldots, x_n$ be as defined in Algorithm \ref{alg1}. Then, with constant probability, we have 

\[ \frac{ \sum_i \max_{j \sim i } (x_i - x_j)^2 }{ \sum_i x_i^2 - \frac{1}{n} \left( \sum_i x_i  \right)^2 } \leq 96\ \sdpval \log d. \]

\end{lemma}

\begin{proof}

We will make use of the following fact that is part of the folkore about Gaussian random variables.
For the sake of completeness, we prove this Fact in \prettyref{app:omitproofs} (\prettyref{fact:appGauss}).

\begin{fact}
\label{fact:Gauss}
Let $Y_1, Y_2, \ldots, Y_d$ be $d$ normal random variables with mean
$0$ and variance at most $\sigma^2$. Let $Y$ be the random
variable defined as $Y \defeq \max \set{Y_i | i\in [d]}$. Then

\[ \Ex{Y} \leq 2\sigma \sqrt{\log d} \]

\end{fact}

Now using this fact we get, 
\[ \Ex{ \max_{j \sim i} (x_j - x_j)^2 }  =  \Ex{ \max_{j \sim i} \inprod{v_i - v_j,g}^2 }  \leq 
2 \max_{j \sim i} \norm{v_j - v_i}^2 \log d.  \]

Therefore, $\Ex{\sum_i \max_{j \sim i} (x_j - x_j)^2  } \leq 2\ \sdpval \log d $.  Using Markov's Inequality we get
\begin{equation} \label{eq:233}
\Pr{ \sum_i \max_{j \sim i} (x_j - x_j)^2 > 48\ \sdpval \log d  }
	\leq  \frac{1}{24}  
\end{equation}

For the denominator, using linearity of expectation, we get 

\[ \Ex{ \sum_i x_i^2 - \frac{1}{n} \left( \sum_i x_i  \right)^2  }  =
	 \sum_i \norm{v_i}^2 - \frac{1}{n} \norm{ \sum_i v_i  }^2. \]

Also recall that the denominator can be re-written as 

$$ \sum_i x_i^2 - \frac{1}{n} \left( \sum_i x_i  \right)^2  =
\frac{1}{n} \sum_{i,j} (x_i-x_j)^2 \mcom $$
which is a sum of squares of gaussians.  Now applying
\prettyref{lem:squaregaussian} to the denominator we conclude

\begin{equation} \label{eq:234}
\Pr{ \sum_i x_i^2 - \frac{1}{n} \left( \sum_i x_i  \right)^2  \geq
	\frac{1}{2} } \geq \frac{1}{12}.  
\end{equation}

Using \eqref{eq:233} and \eqref{eq:234} we get that 

\[ \Pr{ \frac{ \sum_i \max_{j \sim i } (x_i - x_j)^2 }{ \sum_i x_i^2 - \frac{1}{n} \left( \sum_i x_i  \right)^2 } 
	\leq 96\  \sdpval \log d  }  > \frac{1}{24} .\]
\end{proof}

\begin{lemma} \label{lem:squaregaussian}
	Suppose $z_1,\ldots, z_m$ are gaussian random variables (not
	necessarily independent) such $ \E [ \sum_i z_i^2 ] = 1$ then
	$$ \Pr{\sum_i z_i^2 \geq \frac{1}{2}} \geq \frac{1}{12} $$
\end{lemma}
\begin{proof}
We will bound the variance of the random variable $R = \sum_i z_i^2$
as follows,
\begin{align*}
	\E [R^2] & = \sum_{i,j} E[z_i^2 z_j^2] \\
	& \leq \sum_{i,j} \left(E[z_i^4]\right)^{\frac{1}{2}} \left(
	E[z_j^4] \right)^{\frac{1}{2}} \\
	& = \sum_{i,j} 3 E[z_i^2] E[z_j^2]  \qquad \textrm{ (Using }
	\E[g^4] = 3 \E[g^2] \textrm{ for gaussians )} \\
	& = 3\left(\sum_{i}  E[z_i^2]\right)^2 = 3 
\end{align*}
By the Paley-Zygmund inequality,
\[ \Pr { R \geq \frac{1}{2} \E[R]} \geq \left(\frac{1}{2}\right)^2
\frac{(\E[R])^2}{\E[R^2]} \geq \frac{1}{12} \mper \]
\end{proof}

We are now ready to complete the proof of  \prettyref{thm:algo}.

\begin{proof}[Proof of \prettyref{thm:algo}]

Let the $x_i$'s be as defined in Algorithm \ref{alg1}.
W.l.o.g, we may assume that\footnote{For any $x \in \R$, $x^+ \defeq \max \set{x,0}$.}
 $  \Abs{  \supp(x^+)  }  < \Abs{ \supp(x^-)  } $.
For each $i \in [n]$, we define $y_i = x_i^+$.

\prettyref{lem:projlinf} shows that with constant probability we have 
\[ \frac{ \sum_i \max_{j \sim i} (x_i - x_j)^2 }{ \sum_i x_i^2 - \frac{1}{n} \left( \sum_i x_i \right)^2} \leq 
96\ \sdpval \log d.  \]

We need to show that 
\[ \frac{ \sum_i \max_{j \sim i} \Abs{y_i^2 - y_j^2} }{ \sum_i y_i^2 - \frac{1}{n} \left( \sum_i y_i \right)^2}
 \leq 6 \sqrt{ \frac{ \sum_i \max_{j \sim i} (x_i - x_j)^2 }{ \sum_i x_i^2 - \frac{1}{n} \left( \sum_i x_i \right)^2} }. \] 

This fact is proved in \cite{bht00}. For the sake of completeness, we
give a proof of this fact in \prettyref{app:omitproofs}
(\prettyref{lem:bht1}). Using \prettyref{lem:projlinf}, we get 

\[ \frac{ \sum_i \max_{j \sim i} \Abs{y_i^2 - y_j^2} }{ \sum_i y_i^2 - \frac{1}{n} \left( \sum_i y_i \right)^2}
 \leq 576 \sqrt{\sdpval \log d}. \]

From \prettyref{lem:levelset} we get that the set output in Step $3$ of Algorithm \ref{alg1} has vertex expansion
at most $576 \sqrt{ \sdpval \log d }$.
\end{proof}

\bibliography{smallset}

\newcommand{\etalchar}[1]{$^{#1}$}
\providecommand{\bysame}{\leavevmode\hbox to3em{\hrulefill}\thinspace}
\providecommand{\MR}{\relax\ifhmode\unskip\space\fi MR }
\providecommand{\MRhref}[2]{%
  \href{http://www.ams.org/mathscinet-getitem?mr=#1}{#2}
}
\providecommand{\href}[2]{#2}
\begin{thebibliography}{AKK{\etalchar{+}}08}

\bibitem[ABS10]{abs10}
Sanjeev Arora, Boaz Barak, and David Steurer, \emph{Subexponential algorithms
  for unique games and related problems}, FOCS, 2010.

\bibitem[AKK{\etalchar{+}}08]{akkstv08}
Sanjeev Arora, Subhash Khot, Alexandra Kolla, David Steurer, Madhur Tulsiani,
  and Nisheeth~K. Vishnoi, \emph{Unique games on expanding constraint graphs
  are easy: extended abstract}, STOC (Richard~E. Ladner and Cynthia Dwork,
  eds.), ACM, 2008, pp.~21--28.

\bibitem[Alo86]{a86}
Noga Alon, \emph{Eigenvalues and expanders}, Combinatorica \textbf{6} (1986),
  no.~2, 83--96.

\bibitem[AM85]{am85}
Noga Alon and V.~D. Milman, \emph{$\lambda_{\mbox{1}}$, isoperimetric
  inequalities for graphs, and superconcentrators}, J. Comb. Theory, Ser. B
  \textbf{38} (1985), no.~1, 73--88.

\bibitem[AMS07]{ams07}
Christoph Amb{\"u}hl, Monaldo Mastrolilli, and Ola Svensson,
  \emph{Inapproximability results for sparsest cut, optimal linear arrangement,
  and precedence constrained scheduling}, FOCS, IEEE Computer Society, 2007,
  pp.~329--337.

\bibitem[AR98]{ar98}
Yonatan Aumann and Yuval Rabani, \emph{An $\bigo{\log k}$ approximate min-cut
  max-flow theorem and approximation algorithm}, SIAM J. Comput. \textbf{27}
  (1998), no.~1, 291--301.

\bibitem[ARV04]{arv04}
Sanjeev Arora, Satish Rao, and Umesh~V. Vazirani, \emph{Expander flows,
  geometric embeddings and graph partitioning}, STOC (L{\'a}szl{\'o} Babai,
  ed.), ACM, 2004, pp.~222--231.

\bibitem[BHT00]{bht00}
Sergey Bobkov, Christian Houdr{\'e}, and Prasad Tetali,
  \emph{$\lambda_{\infty}$ vertex isoperimetry and concentration},
  Combinatorica \textbf{20} (2000), no.~2, 153--172.

\bibitem[Bor75]{b75}
Christer Borell, \emph{The brunn-minkowski inequality in gauss space},
  Inventiones Mathematicae \textbf{30} (1975), no.~2, 207--216.

\bibitem[FHL08]{fhl08}
Uriel Feige, MohammadTaghi Hajiaghayi, and James~R. Lee, \emph{Improved
  approximation algorithms for minimum weight vertex separators}, SIAM J.
  Comput. \textbf{38} (2008), no.~2, 629--657.

\bibitem[IM12]{im12}
Marcus Isaksson and Elchanan Mossel, \emph{New maximally stable gaussian
  partitions with discrete applications}, To Appear in Israel Journal of
  Mathematics, 2012.

\bibitem[LLR95]{llr95}
Nathan Linial, Eran London, and Yuri Rabinovich, \emph{The geometry of graphs
  and some of its algorithmic applications}, Combinatorica \textbf{15} (1995),
  no.~2, 215--245.

\bibitem[LR99]{lr99}
Frank~Thomson Leighton and Satish Rao, \emph{Multicommodity max-flow min-cut
  theorems and their use in designing approximation algorithms}, J. ACM
  \textbf{46} (1999), no.~6, 787--832.

\bibitem[RS10]{rs10}
Prasad Raghavendra and David Steurer, \emph{Graph expansion and the unique
  games conjecture}, STOC (Leonard~J. Schulman, ed.), ACM, 2010, pp.~755--764.

\bibitem[RST12]{rst12}
Prasad Raghavendra, David Steurer, and Madhur Tulsiani, \emph{Reductions
  between expansion problems}, IEEE Conference on Computational Complexity,
  IEEE, 2012, pp.~64--73.

\bibitem[RT12]{rt12}
Prasad Raghavendra and Ning Tan, \emph{Approximating csps with global
  cardinality constraints using sdp hierarchies}, SODA (Yuval Rabani, ed.),
  SIAM, 2012, pp.~373--387.

\bibitem[ST78]{st78}
Vladimir~N Sudakov and Boris~S Tsirel'son, \emph{Extremal properties of
  half-spaces for spherically invariant measures}, Journal of Mathematical
  Sciences \textbf{9} (1978), no.~1, 9--18.

\bibitem[ST12]{st12}
David Steurer and Prasad Tetali, \emph{Personal communication}.

\end{thebibliography}
\bibliographystyle{amsalpha}


\appendix

\section{Reduction between Vertex Expansion and Symmetric Vertex Expansion}
\label{sec:symv}

In this section we show that the computation of the vertex expansion is essentially 
equivalent to the computation of symmetric vertex expansion. Formally, we prove the following theorems.  

\begin{theorem}
\label{thm:symv}

Given a graph $G = (V, E)$, there exists a graph $H$ such that $\max_{i \in V(H)} \deg(i) \leq  \left( \max_{i \in V(G)} \deg(i) \right)^2+
\max_{i \in V(G)} \deg(i)$ such that 
\[ \phivs(G) \leq \phiv(H) \leq \frac{\phivs(G) }{1 - \phivs(G)} \mper \]

\end{theorem}

\begin{proof}

Let $G^2$ denote the graph on $V(G)$ that corresponds to two hops in the graph $G$. Formally,
\[ \set{u,v} \in E(G^2) \iff \exists w \in V(G), (u, w) \in E(G) \textrm{ and } (w, v) \in E(G) \mper \]
Let $H = G \cup G^2$, i.e., $V(H) = V(G)$ and $E(H) = E(G) \cup E(G^2)$.

Let $S \subset V(G)$ be a set with small symmetric vertex expansion $\phivs(S) = \e$. 
Let $S' = S - N_G(\bar{S})$ be the set
of vertices obtained from $S$ by deleting it's internal boundary. It is easy to see that
\[ N_H(S') = N_G(S) \cup N_G(\bar{S}) \mper \]
Moreover, since $N_G(\bar{S}) \leq \phivs(S) w(S)$ we have $w(S') \geq w(S)(1 - \phivs_G(S ))$. Hence the vertex expansion
of the set $S'$ is upper-bounded by,
\[ \phiv_H(S') \leq \frac{\phivs_G(S) }{1 - \phivs_G(S)} \mper \]

Conversely, suppose $T \subset V(H)$ be a set with small vertex expansion $\phiv_H(T) = \e$. Consider the set
$T' = T \cup N_G(T)$. Observe that the internal boundary of $T'$ in the graph $G$ is given by $N_G(\bar{T}') = N_G(T)$.
Further the external boundary of $T'$ is given by $N_G(T') = N_G(N_G(T)) = N_{G^2}(T)$. Therefore, we have
\[N_G(T') \cup N_G(\bar{T}') = N_G(T) \cup N_{G^2}(T) = N_H(T) \mper \]
Further since $w(T') \geq w(T)$, we have $\phivs_G(T') \leq \phiv_H(T)$.

This completes the proof of the Theorem.

\end{proof}

\begin{theorem}
\label{thm:symv2}
Given a graph $G$, there exists a graph $G'$ such that $\max_{i \in V(G)} \deg(i) = \max_{i \in V(G')} \deg(i)$
and $\phiv(G) = \Theta(\phivs(G'))$.
Moreover, such a $G'$ can be computed in time polynomial in the size of $G$.  
\end{theorem}

\begin{proof}

Given graph $G$, we construct $G'$ as follows. 
We start with $V(G') = V(G) \cup E(G)$, i.e., $G'$ has a vertex for each vertex in $G$ and for each edge
in $G$. For each edge $\set{u,v} \in E(G)$, we add edges $\set{u,\set{u,v}}$ and $\set{v,\set{u,v}}$ in $G'$.  
For a vertex $i \in V(G) \cap V(G')$, we set its weight to be $w(i)$.
For a vertex $\set{u,v} \in E(G) \cap V(G')$, we set its weight to be $\min \set{ w(u)/\deg(u), w(v)/\deg(v)}$.

It is easy to see that $G'$ can be computed in time polynomial in the size of $G$, and that 
$\max_{i \in V(G)} \deg(i) = \max_{i \in V(G')} \deg(i)$.

We first show that $\phiv(G) \geq \phivs(G')/2$. Let $S \subset V(G)$ be the set having the
least vertex expansion in $G$. Let 
\[ S' = S \cup \set{ \set{u,v} | \set{u,v} \in E(G) \textrm{ and } u \in S \textrm{ or } v \in S } \mper \] 
By construction, we have $w(S) \leq w(S')$,  $N_G(S) = N_{G'}(S')$ and 
\[ w( N_{G'}(\bar{S}') ) \leq \sum_{u \in N_{G'}(S')} \deg(u) \frac{ w(u) }{\deg(u)}  \leq w(N_{G'}(S')) \mper  \]
Therefore, 
\[ \phivs(G') \leq \phivs_{G'}(S') = \frac{ w(N_{G'}(S')) + w(N_{G'}(\bar{S}'))   }{w(S')} 
	\leq \frac{2 w(N_G(S)) }{  w(S) } = 2 \phiv_G(S) = 2 \phiv(G) \mper \]


Now, let $S' \subset V(G')$  be the set having the least value of $\phivs_{G'}(S')$ and let $\e = \phivs_{G'}(S')$.
We construct the set $S$ as follows. We let $S_1 = S' \backslash N_{G'}(\bar{S}')$, i.e. we obtain $S_1$
from $S'$ by deleting it's internal boundary. Next we set $S = S_1 \cap V(G)$.
More formally, we let $S$ be the following set.
\[ S = \set{ v \in S' \cap V(G) | v \notin N_{G'}(\bar{S}')} \mper   \]
By construction, we get that $N_G(S) \subseteq N_{G'}(S') \cup N_{ G'}(\bar{S}')$.
Now, the internal boundary of $S'$ has weight at most $\e w(S')$. 
Therefore, we have 
\[ w(S_1) \geq (1 - \e)w(S') \mper \]
We need a lower bound on the weight of the set $S$ we constructed. To this end, we make the  
following observation. For each
vertex $\set{u,v} \in S_1 \cap E(G)$, $u$ or $v$ also has to be in $S_1$ (If not, then deleting $\set{u,v}$
from $S'$ will result in a decrease in the vertex expansion thereby contradicting the optimality of the choice 
of the set $S'$). 
Therefore, we have the following
\[ \sum_{ \set{u,v} \in S_1 \cap E(G)} w( \set{u,v} ) =  
\sum_{ \set{u,v} \in S_1 \cap E(G)} \min \set{ \frac{w(u)}{\deg(u)}, \frac{w(u)}{\deg(u)}   } \leq
\sum_{u \in S_1 \cap V(G)}  w(u) = w(S) \mper \]
Therefore,
\[ w(S) \geq \frac{w(S_1)}{2} \geq (1 - \e) \frac{w(S')}{2}   \]
Therefore, we have 
\[   \phiv(G) \leq \phiv_G(S) = \frac{w(N_G(S))}{w(S)} \leq \frac{ w(N_{G'}(S') \cup N_{ G'}(\bar{S}')  }{(1 - \e) w(S')/2} 
	= 4 \phivs_{G'}(S') = 4 \phivs(G') \mper  \]

Putting these two together, we have 
\[  \frac{\phiv(G)}{2} \leq \phivs(G') \leq 4 \phiv(G)   \mper \]

\end{proof}

\section{Omitted Proofs}
\label{app:omitproofs}

\begin{lemma}[\cite{bht00}]
\label{lem:bht1}

Let $z_1, \ldots, z_n \in R$ and let $x_i \defeq z_i^+$. Then 
\[ \frac{ \sum_i \max_{j \sim i} \Abs{x_i^2 - x_j^2} }{ \sum_i x_i^2 - \frac{1}{n} \left( \sum_i x_i \right)^2}
  \leq 6 \sqrt{ \frac{ \sum_i \max_{j \sim i} (z_i - z_j)^2 }{ \sum_i z_i^2 - \frac{1}{n} \left( \sum_i z_i \right)^2} }. \] 
\end{lemma}

\begin{proof}

W.l.o.g we may assume that $\Abs{\supp(Z^+)} = \Abs{\supp(Z^-)} = \ceil{n/2} $ and that $z_1 \geq z_2  \geq \ldots \geq z_n $.

Note that for any $i \in [n]$, we have  $ \max_{ j \sim i, \& j < i} (z_j^+ - z_i^+)^2 +  
\max_{ j \sim i, \& j > i} (z_j^- - z_i^-)^2 \leq 2 \max_{j \sim i} (z_j - z_i)^2  $.
Now,

\begin{eqnarray*}
 \frac{ \sum_i \max_{j \sim i} (z_j - z_i)^2 }{  \sum_i z_i^2} 
 & \geq &  \frac{ \sum_{ i } \max_{j < i \& j \sim i} (z_j^+ - z_i^+)^2 + \sum_{ i } \max_{j > i \& j \sim i} (z_j^- - z_i^-)^2 }
		{2\left( \sum_{ i \in \supp(Z^+)} z_i^2  +  \sum_{ i \in \supp(Z^-)} z_i^2 \right) } \\
 & \geq  &  \min \left\{ \frac{  \sum_{ i} \max_{j <i \& j \sim i} (z_j^+ - z_i^+)^2  }{ 2  \sum_{ i \in \supp(Z^+)} z_i^2 }, 
	\frac{  \sum_{ i} \max_{j > i \& j \sim i} (z_j^- - z_i^-)^2  }{ 2 \sum_{ i \in \supp(Z^-)} z_i^2 }  \right\}
\end{eqnarray*}

W.l.o.g we may assume that \[
\frac{  \sum_{ i } \max_{j < i \& j \sim i} (z_j^+ - z_i^+)^2 }{ \sum_{ i \in \supp(Z^+)} z_i^2 } \leq 
\frac{  \sum_{ i } \max_{j > i \& j \sim i} (z_j^- - z_i^-)^2 }{ \sum_{ i \in \supp(Z^-)} z_i^2 }
\]

\[ \frac{  \sum_{ i} \max_{j \sim i} (x_j - x_i)^2 }{ \sum_i x_i^2 } \leq 2  \frac{ \sum_i \max_{j \sim i} (z_j - z_i)^2 }{  \sum_i z_i^2  } \]

We have

\begin{eqnarray*}
 \max_{j \sim i, j < i}  (x_j^2  - x_i^2 )  &=& 
		\max_{j \sim i, j < i}  (x_j - x_i )(x_j + x_i ) \\
 & \leq & \max_{j \sim i, j < i}  \left(  (x_j - x_i)^2 +  2 x_i  (x_j - x_i) \right) \\
 & \leq & \max_{j \sim i, j < i} (x_j - x_i)^2  + 2 x_i \max_{j \sim i, j < i}  (x_j - x_i) \\ 
& \leq &  \sum_i \max_{j \sim i, j < i} (x_j - x_i)^2 + 2 \sqrt{ \sum_i x_i^2  }  \sqrt{\max_{j \sim i, j < i}  (x_j - x_i)^2 } 
	\qquad \textrm{ Cauchy-Schwarz} \\
 & = & \linf \sum_i x_i^2 + 2 \sqrt{\linf} \sum_i x_i^2 \\ 
\end{eqnarray*}

Thus we have 
\[  \frac{\sum_i  \max_{j \sim i, j < i}  (x_j^2 - x_i^2) }{ \sum_i x_i^2}  
\leq 6 \sqrt{ \frac{ \sum_i \max_{j \sim i} (z_j - z_i)^2 }{  \sum_i z_i^2  } } \]
\end{proof}

\begin{lemma}
\label{lem:bal-vert-to-vert}
A \cvss hardness for $b$-Balanced-vertex expansion implies a 
\cvs{2}{2} hardness for vertex expansion.
\end{lemma}

\begin{proof}
Fix a graph $G = (V,E)$.

\paragraph{Completeness}
If $G$ has Balanced-vertex expansion at most $c$, then clearly its vertex expansion is also at most $c$.

\paragraph{Soundness}
Suppose we have a polynomial time algorithm that outputs a set $S$ having $\phiv(S) \leq s$
whenever $G$ has a set $S'$ having $\phiv(S') \leq 2c$. Then this 
algorithm  can be used as an oracle to find a balanced set of vertex expansion at most $s$. This would
contradict the hardness of Balanced-vertex expansion.

First we find a set, say $T$, having $\phiv(T) \leq s$. 
If we are unable to find such a $T$, we stop. If we find such a set $T$ and $T$ has balance at least $b$, then
we stop. Else, we delete the vertices in $T$ from $G$ and repeat.
We continue until the number of deleted vertices first exceeds a $b/2$ fraction of the vertices.

If the process deletes less than $b/2$ fraction of the vertices, then the remaining graph (which has
at least $(1 - b/2)n$ vertices) has conductance $2c$, and thus in the original graph every $b$-balanced cut has
conductance at least $c$. This is a contradiction !

If the process deletes between $b/2$ and $1/2$ of the nodes, then the union of the deleted sets gives
a set $T'$ with $\phiv(T') \leq s$ and balance of $T'$ at least $b/2$. 
%
\end{proof}

\begin{fact}
\label{fact:appGauss}
Let $Y_1, Y_2, \ldots, Y_d$ be $d$ standard normal random variables. Let $Y$ be the random
variable defined as $Y \defeq \max \set{Y_i | i\in [d]}$. Then
\[ \Ex{Y^2} \leq 4 \log d \qquad \textrm{ and } \qquad \Ex{Y} \leq 2 \sqrt{\log d} \mper \]
\end{fact}

\begin{proof}

For any $Z_1, \ldots, Z_d \in \R$ and any $p \in \Z^+$, we have $\max_i \Abs{Z_i} \leq (\sum_i Z_i^p)^{\frac{1}{p}}$.
Now $Y^2 = (\max_i X_i)^2 \leq \max_i X_i^2$.

\begin{eqnarray*}
\Ex{Y^2} & \leq & \Ex{ \left( \sum_i X_i^{2p} \right)^{\frac{1}{p}} }
  \leq  \left(\Ex{ \sum_i X_i^{2p} } \right)^{\frac{1}{p}} \quad \textrm{ ( Jensen's Inequality )} \\
 & \leq & \left( \sum_i \left( \Ex{X_i^2} \right) \frac{(2p)! }{(p)! 2^{p} } \right)^{\frac{1}{p}}
  \leq  2 p d^{\frac{1}{p}} \quad \textrm{(using $(2p)!/p! \leq (2p)^{p} $ )} \\
\end{eqnarray*}

Picking $p = \log d$ gives $\Ex{Y^2} \leq 2e \log d$.

Therefore $\Ex{Y} \leq \sqrt{\Ex{Y^2}} \leq \sqrt{2 e \log d} $.

\end{proof}

\section{Noise Operators}

Let $H$ be a Markov chain and let $F : V(H^k) \to \set{0,1} $ be any boolean function.
In this section we prove some basic properties of $\Gamma_{1 - \eta}F$.
We restate the definition of our Noise Operator $\Gamma_{1 - \eta}$.
\[ \Gamma_{1 - \eta}F(X) = (1 - \eta)F(X) + \eta \E_{ Y \sim X} F(Y)  \]

The Fourier expansion of the function $F$ is $ F = \sum_\sigma \hat{f}_{\sigma} e_{\sigma}$
where $\set{e_{\sigma}}$ is the set of eigenvectors of $H^k$. It is easy to see that 
$e_{\sigma} = e_{\sigma_1} \otimes \ldots \otimes e_{\sigma_k}$, where the $\set{e_{\sigma_i}}$
are the eigenvectors of $H$.

\begin{lemma}
\label{lem:lowdeg}
(Decay of High degree Coefficients)
Let $Q_j$ be the multi-linear polynomial representation of $ \Abs{ \Gamma_{1 - \eta}F(X) - \Gamma_{1 - \eta}F(Y_j)}$.
Then, 
\[ \var(Q_j^{>p}) \leq (1 - \e \eta)^{2p} \]

\end{lemma}

\begin{proof}

\begin{eqnarray*}
\Gamma_{1 - \eta} F(X) & = &  (1 - \eta)F(X) + \eta \E_{ Y \sim X} F(Y)  \\
 & = & \sum_{\sigma} \hat{f}_{\sigma} \Ex{e_{\sigma}(X) + \E_{Y \sim X}F(Y) } \\
 & = & \sum_{\sigma} \hat{f}_{\sigma} \Pi_{i \in \sigma} \left( (1 - \eta)e_{\sigma_i}(X_i) + \E_{Y_i \sim X_i} e_{\sigma_i}(Y_i) \right) \\
\end{eqnarray*}

We bound the second moment of $\Gamma_{1 - \eta}F$ as follows

\begin{eqnarray*}
\E_{X} \left( \Gamma_{1 - \eta} F(X) \right)^2  
  & = & \sum_{\sigma} \hat{f}_{\sigma}^2 \E_{X} \Pi_{i \in \sigma} \left( (1 - \eta)e_{\sigma_i}(X_i) + \eta \E_{Y_i \sim X_i} e_{\sigma_i}(Y_i) \right)^2 \\
  & = & \sum_{\sigma} \hat{f}_{\sigma}^2  \Pi_{i \in \sigma} \left( (1 - \eta)^2 \E_{X_i} e_{\sigma_i}(X_i)^2 +
	 \eta^2 \E_{X_i} \left( \E_{Y_i \sim X_i} e_{\sigma_i}(Y_i) \right)^2
		+ 2 \eta (1 - \eta) \E_{X_i} \E_{Y_i \sim X_i} e_{\sigma_i}(X_i) e_{\sigma_i}(Y_i)  \right)^2 \\	
  & = & \sum_{\sigma} \hat{f}_{\sigma}^2  \Pi_{i \in \sigma} \left( (1 - \eta)^2 + \eta^2 \lambda_i^2 + 2 \eta (1 - \eta) \lambda_i \right) \\
  & = & \sum_{\sigma} \hat{f}_{\sigma}^2  \Pi_{i \in \sigma} \left( 1 - \eta + \eta \lambda_i \right)^2 \\
\end{eqnarray*}

Therefore,

\begin{eqnarray*}
\var(Q_j^{>p}) & \leq & 4 \sum_{\sigma : \Abs{\sigma} > p} \hat{f}_{\sigma}^2  \Pi_{i \in \sigma} \left( 1 - \eta + \eta \lambda_i \right)^2 \\
 & \leq & \sum_{\sigma : \Abs{\sigma} > p} \hat{f}_{\sigma}^2  \left( 1 -  \e \eta \right)^{2 \Abs{\sigma}  } \\ 
 & \leq & ( 1 - \e \eta  )^{2 p} \\
\end{eqnarray*}

Here the second inequality follows from the fact that all non-trivial eigenvalues of $H$ are at most $1 - \e$
and the third inequality follows Parseval's indentity. 
\end{proof}

\end{document}